\documentclass[letterpaper]{article} 
\usepackage{aaai24_arxiv}  
\usepackage{times}  
\usepackage{helvet}  
\usepackage{courier}  
\usepackage[hyphens]{url}  
\usepackage{graphicx} 
\usepackage{natbib}  
\usepackage{caption} 
\usepackage{algorithm}
\usepackage{algpseudocode}
\usepackage{newfloat}
\usepackage{listings}
\DeclareCaptionStyle{ruled}{labelfont=normalfont,labelsep=colon,strut=off} 
\floatstyle{ruled}
\newfloat{listing}{tb}{lst}{}
\floatname{listing}{Listing}
\usepackage[utf8]{inputenc}
\usepackage{amsfonts}
\usepackage{amssymb}
\usepackage{amsmath}
\usepackage{mathtools}
\usepackage{forest}
\usepackage{subfig}
\usepackage{etoolbox}
\usetikzlibrary{shapes.geometric,positioning, arrows,trees, positioning, fit, calc}
\tikzset{circ/.style={circle,draw, solid, inner sep=0pt,minimum size=8mm}}
\usepackage{amsthm}
\usepackage[hidelinks]{hyperref} 
\usepackage[capitalise]{cleveref}
\usepackage{thmtools,thm-restate}
\usepackage{comment}
\usepackage{xspace}

\newtheorem{theorem}{Theorem}

\theoremstyle{definition}
\newtheorem{definition}[theorem]{Definition}
\newtheorem{example}[theorem]{Example}
\newtheorem{corollary}[theorem]{Corollary}
\newtheorem{lemma}[theorem]{Lemma}

\newtheorem{innercustomex}{Example}
\newenvironment{customex}[1]
  {\renewcommand\theinnercustomex{#1}\innercustomex}
  {\endinnercustomex}

\crefname{definition}{Definition}{Definitions}
\crefname{figure}{Figure}{Figures}
\let\vec\mathbf

\newcommand{\Maj}{\mathit{Maj}}
\newcommand{\LIQUID}{\textsc{Liquid}}
\newcommand{\calI}{\mathcal{I}}
\newcommand{\calL}{\mathcal{L}}
\newcommand{\calF}{\mathcal{F}}
\newcommand{\agg}{\mathit{agg}}
\newcommand{\MinMax}{\textsc{MinMax}\xspace}
\newcommand{\MinSum}{\textsc{MinSum}\xspace}
\newcommand{\LexiMin}{\textsc{LexiMin}\xspace}

\title{Unravelling Expressive Delegations: Complexity and Normative Analysis}

\author {
    Giannis Tyrovolas\textsuperscript{\rm 1},
    Andrei Constantinescu\textsuperscript{\rm 2},
    Edith Elkind\textsuperscript{\rm 3}
}
\affiliations {
    \textsuperscript{\rm 1}Independent\\
    \textsuperscript{\rm 2}ETH Zurich\\
    \textsuperscript{\rm 3}University of Oxford\\
    johntyro@hotmail.com, aconstantine@ethz.ch, elkind@cs.ox.ac.uk
}

\begin{document}

\maketitle

\begin{abstract}
We consider binary group decision-making under a rich model of liquid democracy recently proposed by \citet{grandi}:~agents submit ranked delegation options, where each option may be a function of multiple agents' votes; e.g., ``I vote yes if a majority of my friends vote yes.'' Such ballots are unravelled into a profile of direct votes by selecting one entry from each ballot so as not to introduce cyclic dependencies. We study delegation via monotonic Boolean functions, and
two unravelling procedures:~\MinSum, which minimises the sum of the ranks of the chosen entries, and its egalitarian counterpart, \MinMax. We provide complete computational dichotomies:~\MinSum is hard to compute (and approximate) as soon as any non-trivial
functions are permitted, and polynomial otherwise; for \MinMax the easiness results extend to arbitrary-arity logical ORs and ANDs taken in isolation, but not beyond. 
For the classic model of delegating to individual agents, we give asymptotically near-tight algorithms for carrying out the two procedures, and efficient algorithms for finding optimal unravellings with the highest vote count for a given alternative. These algorithms inspire novel tie-breaking rules for the setup of voting to change a status quo. We then introduce a new axiom, which can be viewed as a variant of the participation axiom, and use algorithmic techniques developed earlier in the paper to show that it is satisfied by \MinSum and a lexicographic refinement of \MinMax (but not \MinMax itself). 
\end{abstract}

\section{Introduction}

Liquid democracy, also referred to as delegative democracy, is a decision-making mechanism that allows the voters more flexibility than representative democracy:~While every voter may vote directly on an issue, voters may also delegate their vote to other voters.
Crucially, delegations are \emph{transitive}: if~Alice delegates to Bob, and Bob delegates to Claire, then Claire votes with the combined power of all three.

Variants of this idea have been explored by many authors; see, e.g., the survey by \citet{origins_of_liquid}. In particular, \citet{ford2002delegative} argues that large-scale direct democracy is infeasible and likely undesirable, whereas
representative democracy is rather rigid: It requires holding elections every so often, with winners representing their entire constituency and losers gaining no representative power.
Further, there is a hard limit on the total number of elected representatives, and every voter can pick from a limited number of candidates.
Delegative democracy can then be seen as a balance of the two, by challenging the premise that the number of representatives needs to be kept small.
In this model, voters can vote directly on issues if they wish to do so. Passive voters can delegate to
representatives called {\em delegates}.
Delegates are not chosen at fixed time points, but need to canvass voters continually,
and acquire the combined power of all voters that delegate to them.
They can vote on issues directly, or delegate to other 
delegates.
Importantly, delegates need not win competitive elections, and their votes are public for the sake of accountability. 
Ford introduces several strengthenings of liquid democracy, such as voters delegating fractions of their votes to different delegates.
A particularly interesting variant is to allow agents to submit ``multiple delegation choices in order of preference.'' This is done partly to deal with the case of cycles, where Alice delegates to Bob, but also Bob delegates to Alice.

\citet{grandi} further extend this model, by allowing each delegation choice to be an arbitrary function of other agents' selections. That is, instead of delegating to Bob, Alice can declare that she supports option X if at least one of Bob, Charlie and Diana supports it (for a formal treatment of this model, see Section~\ref{section:prelims}).
An agent's ballot is then a ranked list of such functions, in order of preference.
This enhancement of the basic model means that cycles can be introduced in much more complicated ways.

Colley et al.~propose two ``unravelling'' procedures to resolve these cycles. \MinSum is a utilitarian method that minimises the sum of preference levels used, and \MinMax is an egalitarian method that minimises the maximum of the preference levels used.
They prove that \MinSum is NP-hard if agents are allowed to use arbitrary-arity logical ANDs and \MinMax is NP-hard for arbitrary Boolean functions.
Also, they give polynomial-time algorithms to unravel instances where agents can only delegate directly to other agents.
\subsection{Our Contribution}
We focus on binary issues and strengthen the complexity results of Colley et al.~to
complete computational dichotomies:~\MinSum is hard to compute (and approximate) as soon as any non-trivial
functions are permitted, and polynomial otherwise; for \MinMax the easiness results extend to arbitrary-arity logical ORs and ANDs taken in isolation, but not beyond. For the standard model of delegation to individual agents, we design algorithms that are faster than those of Colley et al.~(indeed, nearly optimal). Then, we consider the problem of determining if there exists an unravelling where a particular alternative is selected, giving efficient algorithms in both cases. For \MinSum, we do this by leveraging Fulkerson’s primal-dual algorithm for computing minimum-cost arborescences \citep{fulkerson}.
Our results suggest novel tie-breaking rules for both unravelling procedures, which are particularly appealing when one of the two input alternatives is seen as a status quo. Interestingly, Fulkerson's ideas also turn out to be useful
from a normative perspective: we put forward a new axiom, which can be seen as an 
adaptation of the classic participation axiom to the setting of liquid democracy, 
and use our algorithmic techniques to show that it is satisfied by \MinSum and a lexicographic refinement of \MinMax, but not \MinMax itself.

\subsection{Related Work}
The formal model of liquid democracy with ranked delegations was put forward by \citet{kotsialou}, who were interested in identifying unravelling procedures
that avoid cycles and have good normative properties. They study the properties of breadth-first search and depth-first search in this context. Their model inspired a number of recent papers \citep{grandi,brill,popular_branchings,utke2023anonymous}; in particular,
the work of Colley et al. extends the model of \citet{kotsialou} to more expressive ballots and is the direct predecessor of our work. \citet{brill} define the concept of delegation graphs, for which unravelling procedures are called delegation rules. They study two important subclasses:~sequence rules and confluent rules; the latter corresponding to selecting an arborescence of the graph. \MinSum and \MinMax are both confluent rules. \citeauthor{brill} provide a wide range of normative and experimental results. Of importance to us, \textsc{MinSum},\footnote{Called BordaBranching in their paper to distinguish it from a different rule called MinSum.} corresponding to minimum-cost arborescences, satisfies two attractive axioms---guru-participation and copy-robustness (the latter is similar to our new axiom)---if tie-breaking is performed lexicographically with respect to a fixed ordering of the agents. In contrast, our axiom is formulated for irresolute \MinSum, and can be satisfied by a resolute anonymous rule (i.e., treating all agents symmetrically), at the expense of
neutrality.
Moreover, building on the popular arborescences framework of \citet{popular_branchings}, the authors experimentally show that on most instances a majority of voters prefer \MinSum arborescences to other arborescences.

There is also a substantial body of work on liquid democracy in the setting
with ground truth \citep{kahng2021liquid,golz2021fluid,caragiannis2019contribution,becker2021can,alouf2022should}, as well as on game-theoretic aspects of delegation 
\citep{bloembergen2019rational,escoffier2020iterative,zhang2021power}.

\section{Preliminaries}
\label{section:prelims}
A finite set of agents/voters $N = \{1, \ldots, n\}$ aim to select an outcome from a finite set $D$ of alternatives. 
For most of our results, we focus on the \emph{binary} setting $D = \{0, 1\}.$

\subsection{Classic Liquid Democracy}\label{sect:classic-liquid}
We first describe the classic model of liquid democracy with ranked delegations; then, in the next section, we will delve into its generalisation proposed by Colley et al.~(\citeyear{grandi}). We will refer to the latter model simply as \emph{liquid democracy} whenever clear from the context. 

The decision-making process proceeds as follows. First, each agent $a \in N$ casts an ordered ballot with size $k_a$: $B_a = B_a(0) \succ \ldots \succ B_a(k_a - 1) \succ \tau_a$,
where $B_a(i) \in N \setminus \{a\}$ for each $i=0, \dots, k_a-1$ and $\tau_a \in D.$ We require  $B_a(i) \neq B_a(j)$ for $0 \leq i < j < k_a.$ For brevity, we write $B_a(k_a) = \tau_a$.  Intuitively, agent $a$ wants to delegate their vote first to $B_a(0),$ as a second option to $B_a(1),$ and so on; their final backup option is to vote directly for $\tau_a \in D.$ When $k_a = 0,$ the ballot of agent $a$ is $B_a = \tau_a,$ which indicates that $a$ will vote directly for $\tau_a.$ After the ballots are cast, the electoral authority \emph{unravels} the preferences by turning the collection of ballots $\vec{B} = (B_a)_{a \in N}$ into a collection of \emph{concrete} votes $\vec{X} = (X_a)_{a \in N} \in D^N.$ Then, a deterministic aggregation rule $\agg : D^N \to D$ is applied to $\vec{X}$ to compute the final decision $\agg(\vec{X}).$ We will sometimes work with \emph{partial} collections of concrete votes $\vec{X} \in (D \cup \{\bot\})^N$ where some entries are unspecified (as indicated by $X_a = \bot$).

To unravel the preferences, one can select, for each agent $a\in N$, a position $0 \leq c_a \leq k_a$ from their ballot and use the respective entry $B_a(c_a)$ to compute the concrete votes. This approach has the attractive feature of explainability:
The collection of positions $\vec{c} = (c_a)_{a \in N}$ forms a \emph{certificate}, which can be reported together with the electoral outcome as an explanation for it. Unravelling procedures that work in this fashion are called \emph{confluent} \citep{brill}. They provide a high degree of accountability, 
as each agent can see whether their direct vote $\tau_a$ was used, or which other agent their vote was delegated to. 

Naturally, not all certificates are acceptable: e.g., if Alice and Bob want to delegate to each other as their first preference, the certificate $\vec{c} = (0, 0)$ leads to a delegation cycle, and hence to no concrete votes being computed for Alice and Bob. 
Following the terminology of Colley et al.~(\citeyear{grandi}), we call such certificates \emph{inconsistent}. It is convenient to formalise this concept in the language of graphs.

A tuple $(N, D, (B_a)_{a \in N})$ can be seen as a directed weighted graph called the \emph{delegation graph}. Its vertex set is  $N \cup D \cup \{r\}$, where $r$ is a special vertex (the {\em root}). There is an edge $d \rightarrow r$ of weight 0 for each $d \in D$ as well as an edge $a \rightarrow B_a(i)$ 
of weight/cost\footnote{We use these terms interchangeably.} 
$i$ for each $a \in A$ and $0 \leq i \leq k_a.$ A certificate $\vec{c}$ is \emph{consistent} if the set of edges $T_\vec{c} = \{(a, B_a(c_a)) \mid a \in A\} \cup \{(d, r) \mid d \in D\}$ is an $r$-\emph{arborescence} (i.e., a directed spanning tree with sink $r$) of the delegation graph. In particular, every node can reach $r$ along edges in $T_\vec{c}$, so there are no delegation cycles. For brevity, we shorten $r$-arborescence to simply arborescence. A consistent certificate $\vec{c}$ induces concrete votes $\vec{X}$ as follows: $X_a=d$ whenever $d$ is on the $a$--$r$ path in $T_\vec{c}$; we say that agent $a \in N$ \emph{votes} $d \in D$ in $\vec{c}$.

When $\vec{B}$ admits multiple consistent certificates, it is natural to prefer ones selecting options that are ranked as highly as possible in the ballots. In particular, one can select the certificates that minimise $\sum_{a \in N}c_a$ or $\max_{a \in N}c_a,$ i.e., maximise the \emph{utilitarian} or \emph{egalitarian} social welfare. We call such certificates \MinSum and \MinMax certificates, and extend this terminology to unravelling procedures that select certificates optimising the respective quantities.
Notice that these procedures are irresolute, as a minimiser of $\sum_{a \in N}c_a$ or $\max_{a \in N}c_a,$ need not be unique. In graph-theoretic terms, \MinSum and \MinMax certificates correspond to, respectively, minimum weight and minimum bottleneck arborescences.
Given a (potentially irresolute) unravelling procedure $F,$ e.g., \MinSum or \MinMax, and ballots $\vec{B},$ we write $F(\vec{B})$ for the set of collections of concrete votes $\vec{X}$ that can be output by $F$ when executed with input $\vec{B}.$ When $F$ is resolute, we simply write $F(\vec{B}) = \vec{X}$ instead of $F(\vec{B}) = \{\vec{X}\}.$

\subsection{Enhanced Liquid Democracy}\label{sect:prelims-enhanced-liquid}

We now present {\em smart voting},\footnote{Not to be confused with the Swiss voting advice application \texttt{www.smartvote.ch}.} as defined by Colley et al.~(\citeyear{grandi}): the extension of the classic delegation model that allows agents to delegate to complex functions of the other agents. For readability, we limit ourselves to the binary setting. We first take a short detour to introduce the notions of Boolean logic needed for our examples and results. 

A \emph{$k$-ary Boolean function} $f \colon \{0, 1\}^k \to \{0, 1\}$ is a mapping from $k$-tuples of binary values to binary values. Function $f$ is \emph{monotonic} if for all $k$-tuples $\vec{x}, \vec{y}$ such that $x_i \leq y_i$ for $1 \leq i \leq k$ it holds that $f(\vec{x}) \leq f(\vec{y}).$ Two $k$-ary Boolean functions $f$ and $g$ are equal if they are \emph{extensionally equal}, i.e., $f(\vec{x}) = g(\vec{x})$ for any $k$-tuple $\vec{x}.$ Given a $k$-ary Boolean function $f,$ substituting values for $\ell$ of its arguments yields a $k'$-ary function for $k' = k - \ell.$ In a slight abuse of notation, we use $0$ and $1$ to denote the two constant functions (of arity zero). We express Boolean functions via {\em Boolean formulas} using connectives such as 
$\land$, $\lor$ and $\neg$;
we often write $\neg x$ as $\overline{x}.$ 
For instance, the ternary majority function $\textsc{Maj}(x, y, z)$ can be written as $(x \land y) \lor (y \land z) \lor (z \land x).$ 
When a Boolean function is monotonic, it can be expressed by a \emph{monotone} Boolean formula, i.e., without negations.

We will work with $n$-ary Boolean functions with arguments indexed by the set of agents $N,$ i.e., functions $f : \{0, 1\}^N \to \{0, 1\}.$ 
When a function does not require some of its arguments, we may omit them, and consider an ``equivalent'' function of smaller arity. Formally, if $f$ only depends on its arguments indexed by $S \subsetneq N$ then we will view it as $f' : \{0, 1\}^S \to \{0, 1\}.$ Conversely, in contexts requiring it, a function $f' : \{0, 1\}^S \to \{0, 1\}$ denotes the function $f : \{0, 1\}^N \to \{0, 1\}$ with redundant arguments added for $N \setminus S.$ Given a (possibly partial) list of concrete votes $\vec{X} \in \{0, 1, \bot\} ^N,$ we write $f(\vec{X})$ for the function obtained from $f$ by substituting $a \mapsto X_a$ for all $a \in N$ with $X_a \neq \bot.$

In the enhanced liquid democracy model of Colley et al.~(\citeyear{grandi}),
for each $a \in N$ and $0 \leq i < k_a$ the $i$-th entry $B_a(i)$ in agent $a$'s ballot is a Boolean function with arguments indexed by $N \setminus \{a\}.$ The final entries $B_a(k_a) = \tau_a$ naturally correspond to one of the two constant functions.

\begin{figure}
    \centering
    \begin{tikzpicture}
        \tikzset{node distance=0.95cm and 1.35cm}
    
        \node[circ] (a){$a$};
        
        \node[circ, right= of a](b) {$b$};
        \node[circ, below= of b](c) {$c$};

        \node[circ, right= of b](d) {$d$};
        
        \node[circ, right= of c](e) {$e$};
        \node[circ, below= of e](f) {$f$};
        \node[circ, below= of f](g) {$g$};

        \node[circ, right= of e](zero) {0};
        \node[circ, below = of zero](one){1};

        \node[draw,  dashed, rounded rectangle, rotate fit = 90, fit={(b)(c)}](elli_bc) {};

        \node[draw, dashed, rounded rectangle, rotate fit = 90, fit={(f)(g)}](elli_fg) {};
        \node[draw, dashed, rounded rectangle, rotate fit = 90, fit={(f)(g)(e)(elli_fg)}](elli_fge) {};

        \draw[->](a) to node[midway,above] {$\lor$} (elli_bc.north);
        \draw[->, dashed](a) to [bend left] (b);

        \draw[->] (b) to [bend left] node[midway,above] {$\neg$}(d);
        \draw[->, dashed] (b) to [bend left] (c);

        \draw[->] (c) to node[midway,above, xshift = 2mm] {$\Maj$}(elli_fge);
        \draw[->, dashed] (c) to [bend left] (a);

        \draw[->] (d.north) to [bend right] (a.north);
        \draw[->, dashed] (d) to (c);
        \draw[->, dotted] (d) to (one);

        \draw[->, dashed] (e) to node[midway,right] {$\land$}(elli_fg);
        \draw[->] (e) to (d);
        \draw[->, dotted] (e) to (zero);

        \draw[->] (f) to [bend left] (c);
        \draw[->, dashed] (f) to (zero);

        \draw[->] (g) to (one);

    \end{tikzpicture}
    \caption{Preference profile in \cref{example:full-model}. Solid lines indicate first preferences, dashed lines indicate second preferences, and dotted lines indicate third preferences. To avoid clutter we have removed the third preferences of $a, b$ and $c$.}
    \label{fig:example}
\end{figure}%
\begin{example}\label{example:full-model}
Consider agents $N = \{a, b, c, d, e, f, g\}$ casting ballots $\vec{B} = (B_a)_{a \in N}$ as follows (\cref{fig:example}): 
\begin{align*}
    B_a &= b \lor c \succ b \succ 0
    &B_b &= \overline{d} \succ c \succ 1 \\
    B_c &= \Maj\left(e, f, g\right) \succ a \succ 1
    &B_d &= a \succ c \succ 1 \\
    B_e &= d \succ f \land g \succ 0
    &B_f &= c \succ 0 \\
    B_g &= 1
\end{align*}
\end{example}

The enhanced model generalises classic liquid democracy in that the latter only allows projection functions (direct delegations) and constant functions (direct votes). Most definitions from the classic model extend in a straightforward way, but a few important points need to be discussed. First, what language is used to express Boolean functions?
Note that the validation check $B_a(i) \neq B_a(j)$ for all $0 \leq i < j \leq k_a$\footnote{Note the last $\leq$ instead of $<.$} has to be replaced with (extensional) function equality: a simple textual comparison would not suffice, as a Boolean function can be expressed by 
multiple Boolean formulas. We want the function-definition language to allow for this check to be performed efficiently, i.e., in polynomial time; hence, the language of arbitrary Boolean formulas is not suitable. Perhaps most importantly, the definition of the 
delegation graph does not extend to the enhanced model, and hence we need a new definition
of a consistent certificate. As the enhanced model generalises the classic model, this definition should coincide with the classic one for instances that only use direct delegations and constants. It will be convenient to use the following reformulation of consistency for the classic model (the proof follows by considering arborescence traversals where vertices occur after their parents).

\begin{lemma} In the classic model, a certificate $\vec{c}$ is consistent iff there exists a linear order $\triangleleft$ on $N$ such that for all agents $a \in N$ it holds that either $B_a(c_a) \in D$ or $B_a(c_a) \triangleleft a.$
\end{lemma}

In other words, a certificate is consistent iff the concrete votes can be computed by starting with $\vec{X} = (\bot, \ldots, \bot)$ and at each step picking an agent $a \in N$ with $X_a = \bot$ and resolving their vote, i.e., setting $X_a$ either to $B_a(c_a) \in D$ or to $X_{B_a(c_a)}$; for the latter choice we need $X_{B_a(c_a)} \neq \bot$, i.e., $B_a(c_a) \triangleleft a,$ which should be read as ``the vote of $B_a(c_a)$ is resolved before that of $a$''. 
Armed with the idea that ``concrete votes can be resolved sequentially,'' we extend the definition of consistency to the enhanced model as follows.

In the enhanced model, we say that a certificate $\vec{c}$ is consistent iff there is a linear order $\triangleleft$ on $N$ that is {\em valid}, in the sense that the concrete votes can be computed in the following way. We start with $\vec{X} = (\bot, \ldots, \bot)$ and go through the agents in order given by $\lhd$. Whenever an agent $a \in N$ is reached, let $f = B_a(c_a)$ be their assigned function in the certificate. Then it should be the case that $f$ simplifies to one of the two constant functions $0$ or $1$ when the values known for agents that appear earlier in $\lhd$ are substituted in. Formally, either $f(\vec{X}) = 0$ or $f(\vec{X}) = 1$ should hold, where equality denotes extensional equality.
We set $X_a \gets f(\vec{X})$ and move on to the next agent in order. That is, the way in which functions are evaluated is not unlike lazy function evaluation in some programming languages. 

For instance, let $N=\{a, b, c\}$. Suppose that the function $f = B_a(c_a)$ of agent $a$ assigned by the certificate is $b \lor c,$ and consider the order 
$b \triangleleft a \triangleleft c$. If, when $a$ is reached in the order, we have $\vec{X} = (\bot, 1, \bot)$, we can compute $f(\vec{X}) = f(\bot, 1, \bot) = 1 \lor c = 1$, even though the vote of $c$ is still unknown. However, for $\vec{X} = (\bot, 0, \bot)$ the order $\lhd$ is not valid: $f(\vec{X}) = f(\bot, 0, \bot) = 0 \lor c = c$ is not extensionally equal to a constant function. 

Crucially, the computed concrete votes $\vec{X}$ do not depend on the choice of the valid ordering $\triangleleft:$ This is because every assignment performed is, in a sense, ``forced.''\footnote{To decide if a valid $\triangleleft$ exists it suffices to attempt to construct it iteratively, by picking, at each step, a new agent $a$ whose vote can be determined based on the previous ones.} Hence, for a consistent certificate, the concrete votes $\vec{X}$ are uniquely determined and are the unique solution to the system of simultaneous equations $X_a = B_a(c_a)(\vec{X})$ for all $a \in N.$ One can ask whether the converse is also true. Namely, given a certificate $\vec{c}$ such that the system of simultaneous equations has a unique solution (i.e., the concrete votes are uniquely determined), is it consistent? If this were the case, this would be an argument in favor of consistent certificates as an attractive ground notion. The next theorem shows that is indeed true when delegations are restricted to monotonic functions (see the appendix for the proof, and all other omitted proofs), and we will argue that there are strong arguments for using monotonic functions from a normative perspective.

\begin{restatable}{theorem}{thfixedpoint}\label{th:fixed-point} Assuming delegation is only permitted to monotonic functions, a certificate $\vec{c}$ is consistent iff the system of simultaneous equations $X_a = B_a(c_a)(\vec{X})$ for all $a \in N$ has a unique solution.
\end{restatable}

We now return to \cref{example:full-model} to illustrate the new definitions, together with \MinSum and \MinMax, whose definitions remain unchanged.

\begin{customex}{\ref{example:full-model}}[continued] 
Observe that $\vec{c} = (0, \ldots, 0)$ is not a consistent certificate. Indeed, the iterative process for this certificate starts with $\vec{X} = (\bot, \ldots, \bot),$ sets $X_g = 1$ in the first step, but can not proceed further. This is because the knowledge of $X_g$ does not immediately uniquely identify any other value in $\vec{X},$ since, considering first preferences only, $g$ is only mentioned in the first preference of $c,$ and $\Maj(e, f, g) = \Maj(e, f, 1)$ is not uniquely determined without knowing $X_e$ and $X_f$ first.

Hence, if there is a consistent certificate using only the first two preferences of all agents, then it is a \MinMax certificate. One such certificate is $\vec{c} = (0, 1, 0, 0, 1, 1, 0).$ To see this, start again the iterative process with $\vec{X} = (\bot, \ldots, \bot).$ We immediately obtain $X_g \gets 1$ and $X_f \gets 0$. Subsequently, $X_e \gets X_f \land X_g = 0 \land 1 = 0.$ Then, $X_c \gets \Maj(X_e, X_f, X_g) = \Maj(0, 0, 1) = 0,$ which could not have been determined before $X_e$ since $\Maj(X_e, 0, 1)=X_e$. Afterwards, $X_b \gets X_c = 0.$ Then $X_a \gets X_b \lor X_c = 0 \lor 0 = 0$ and $X_d \gets X_a = 0.$ So, the iterative procedure has successfully determined the concrete votes $\vec{X} = (0, 0, 0, 0, 0, 0, 1),$ so $\vec{c}$ is a \MinMax certificate; i.e., $\max_{a \in N} c_a = 1$. Note that the order used to resolve the votes is $g \triangleleft f \triangleleft e \triangleleft c \triangleleft b \triangleleft a \triangleleft d.$ Notationally, the concrete votes satisfy $\vec{X} \in \MinMax(\vec{B}).$

The previous certificate had $\sum_{a \in N} c_a = 3.$ In contrast, a \MinSum certificate is given by $\vec{c} = (0, 0, 2, 0, 0, 0, 0),$ which satisfies $\sum_{a \in N} c_a = 2.$ Performing the unravelling in this case would yield the direct votes $\vec{X} = (1, 0, 1, 1, 1, 1, 1) \in \MinSum(\vec{B}).$
Note that $\vec{c}$ is not the only \MinSum certificate:~$\vec{c'} = (0, 0, 0, 2, 0, 0, 0)$ would be a different one.
\end{customex}

We end these preliminaries by returning to the expressivity of the language used to specify functions. For computational tractability, all functions that appear in the ballots should be polynomial-time computable. Moreover,  ballot validation, as well as checking whether a given certificate is consistent (say after it has been reported to the voters) both require that checking extensional function equality can be achieved in polynomial time. These are all rather demanding properties, and they notably fail unless P = NP even for functions expressed as monotone formulas in disjunctive normal form (DNF), as reported by Colley et al.~(\citeyear{grandi}). As a resolution, they propose restricting the class of allowable functions to minimised DNF formulas, for which the problems become tractable. We do not make this assumption, but instead note that for our results these conditions will be either trivially satisfied or not necessary (the latter holds for the hardness parts of dichotomy results).

\section{Complexity of the Enhanced Model}
\label{section:Complexity}

In this section, we provide complete computational dichotomies for \MinSum and \MinMax for the binary setting assuming that all allowable functions are monotonic. It is easy to see why non-monotonic delegations may cause issues that are more fundamental than computational complexity. E.g., suppose the supporters of a substantial-sized party delegate to the negation of the vote of a high-influence individual who votes directly for 1. This individual is then motivated to vote directly for 0 knowing that this flips the vote of many agents from $0$ to $1$, adding an undesirable strategic element to the setup. 

For $F \in \{\MinSum, \MinMax\}$ and a class of Boolean functions $\calF$ including projections and constants (i.e., direct delegations and direct votes), we write $F_\calF$ for $F$ restricted to delegations to functions in $\calF.$ Notable examples of $\calF$ include \textsc{Bool}, the class of all Boolean functions, \textsc{Mon}, the class of all monotonic Boolean functions, \textsc{Liquid}, which only allows projections and constants (i.e., classic liquid democracy), $\textsc{Or}_2$, which additionally allows binary disjunctions (similarly $\textsc{And}_2$), $\textsc{Or},$ which allows arbitrary-arity disjunctions (similarly $\textsc{And}$) and finally $\textsc{OrAnd}_2 = \textsc{Or}_2 \cup \textsc{And}_2.$ We assume $\calF$ to be invariant to permuting agents in $N;$ e.g., if agent $a$ can delegate to $(b \lor c) \land d,$ then $a'$ can delegate to $(b' \lor c') \land d'$ for any pairwise distinct $a', b', c', d'.$ This assumption is mild given that anonymity is seen as a desirable 
preference aggregation axiom.

For brevity, we adopt the following notational conventions. We say that $\MinSum_\calF$ is NP-hard if the decision problem ``Given $x$, decide if there is a consistent certificate with $\sum_{a \in N}c_a \leq x$'' is NP-hard. We say that it is poly-time computable if there exists a polynomial-time algorithm outputting a $\MinSum_\calF$ certificate. We say that $\MinSum_\calF$ is hard to approximate within a constant factor if, unless P = NP, for no constant $C \geq 1$ does there exist a polynomial-time algorithm outputting a consistent certificate $\vec{c}$ with $\sum_{a \in N}c_a$ at most $C$ times more than this sum for a \MinSum certificate. Moreover, we say that it is hard to approximate within an $n^{1 - \epsilon}$ factor if for no $\epsilon > 0$ the above holds for $C = n^{1 - \epsilon}.$ We extend this terminology to $\MinMax_\calF.$

Colley et al.~(\citeyear{grandi}) show that $\MinSum_\LIQUID$ and $\MinMax_\LIQUID$ are polynomial-time computable
while $\MinSum_\textsc{And}$ and $\MinMax_\textsc{Bool}$ are \textsc{NP}-hard. We strengthen these results by giving a full characterisation of when the problems admit polynomial-time algorithms.

\paragraph{Complexity of $\MinSum_\calF$.}

We start by showing that if either $\textsc{Or}_2 \subseteq \calF$ or $\textsc{And}_2 \subseteq \calF,$ then $\MinSum_\calF$ is NP-hard and hard to approximate within an $n^{1 - \epsilon}$ factor, even for ballots of maximum size two. 

\begin{restatable}{theorem}{thmorandtwohardtopolyapp}\label{thm:or2-and2-hard-to-poly-app} 
$\MinSum_{\textsc{Or}_2}$ and $\MinSum_{\textsc{And}_2}$ are hard to $n^{1 - \epsilon}$-factor approximate even for maximum ballot size 2.
\end{restatable}

We prove this result by reducing from the decision version of \textsc{VertexCover}. We construct voter gadgets for each vertex and edge. Edge gadgets can be resolved using the voters' first preferences if and only if the edge is covered.
We then use gap amplification to show hardness of approximation. The $n^{1 - \epsilon}$ bound is tight for ballots of size at most $2$ since there exists a simple $n$-approximation: unless everyone can get their first (cost $0$) choice, the optimal solution has cost at least $1$, so giving everyone their second choice is an $n$-approximation.

\paragraph{Complexity of $\MinMax_\calF$.}

Next, we show that if $\textsc{OrAnd}_2 \subseteq \mathcal{F}$ then $\MinMax_\calF$ is NP-hard even for ballots of maximum size three,\footnote{Note the contrast with $\MinSum_\calF,$ where two sufficed for hardness. For maximum size two, finding a \MinMax certificate reduces to checking whether $\vec{c} = (0, \ldots, 0)$ is consistent, whose difficulty is given solely by the difficulty of testing function equivalence, which is easy for $\textsc{OrAnd}_2.$} and hard to approximate within a constant factor. In contrast, we prove that $\MinMax_\textsc{Or}$ and $\MinMax_\textsc{And}$ are polynomial-time computable, showing the need for both $\land$ and $\lor$ to obtain hardness.
We begin with the easiness result.

\begin{restatable}{theorem}{thmminmaxbooleanpolytime}\label{thm:minmax-Boolean-poly-time} $\MinMax_\textsc{Or}$ and $\MinMax_\textsc{And}$ are polynomial-time computable (for arbitrary ballot sizes).
\end{restatable}

We now turn to the hardness results. The proofs are similar in spirit to those for \MinSum. We reduce from a version of \textsc{3SAT}.

\begin{restatable}{theorem}{thhardnessminmax}\label{th:hardness-min-max} $\MinMax_{\textsc{OrAnd}_2}$ is NP-hard even for maximum ballot size 3, 
and hard to constant-factor approximate.
\end{restatable}

\paragraph{Computational Dichotomies.}

So far, we have proven results concerning $\lor$ and $\land$. We now extend our results to full computational dichotomies establishing tractability/intractability for every class of monotonic functions $\calF \subseteq \textsc{Mon}.$ We implicitly assume $\LIQUID \subseteq \calF$ since $\LIQUID$ alone entails tractability anyway. 

\begin{restatable}{theorem}{thmdich}\label{thm:dich} Assume $\calF \subseteq \textsc{Mon}.$ If $\calF = \LIQUID,$ $\MinSum_\calF$ is poly-time computable. Otherwise, 
$\MinSum_\calF$ is NP-hard and hard to $n^{1 - \epsilon}$-approximate even for ballots of size at most 2. Also, If $\calF \subseteq \textsc{Or}$ or $\calF \subseteq \textsc{And},$ $\MinMax_\calF$ is poly-time computable. Otherwise, $\MinMax_\calF$ is NP-hard even for ballots of size at most 3, and hard to constant-factor approximate.
\end{restatable}

As an immediate corollary, if agents are allowed to delegate to some odd majority then \MinSum and \MinMax are hard, resolving a question left open by Colley et al.~(\citeyear{grandi}).

\section{Complexity of the Classic Model}

We now revisit the classic model of liquid democracy with ranked delegations, as described in \cref{sect:classic-liquid}. In particular, we study the \MinSum and \MinMax unravelling procedures, corresponding to minimum cost and minimum bottleneck arborescences in the delegation graph.

\subsection{Efficient Computation of \MinSum and \MinMax Certificates}\label{sect:efficient-minsum-minmax}

Our first results concern efficiently computing an arbitrary optimal certificate. Compared to previous works, our algorithm for \MinMax is asymptotically tight, while the one for \MinSum is tight up to a logarithmic factor. In this section only, we do not make the assumption that $D=\{0, 1\}$, since our analysis applies to general domains with no additional complications. We work with a modified delegation graph: we remove the vertices that correspond to $D \subseteq V$, and redirect the edges pointing to them to point directly to the root $r.$\footnote{This gives slightly tighter and cleaner time bounds, whilst not impacting the certificate-arborescence correspondence.} We then have $|V| = n + 1$; we write $m = |E| = n + \sum_{a \in N}k_a$ for the total number of edges and $\ell$ for the maximum delegation ballot size. The algorithms proposed by \citet{grandi} run in time $\mathcal{O}(nm)$ for \textsc{MinSum} and $\mathcal{O}(n^2\ell^2)$ for \textsc{MinMax}.

\begin{restatable}{theorem}{thmefficientminsumminmax}\label{thm:efficient-minsum-minmax} 
A \MinSum certificate can be computed in time $O(n \log n + m)$, and
a \MinMax certificate can be computed in time $O(n + m)$.
\end{restatable}

The algorithm we propose for \MinMax can be described in purely graph-theoretic terms:~given a directed graph with integer edge weights between $0$ and $W,$ it computes a minimum bottleneck arborescence with a given root in time $O(n + m + W).$ Despite the simplicity of the algorithm, we have not been able to find it in prior work.

\subsection{Control for \MinSum and \MinMax}

Thus far, we have concerned ourselves with the irresolute variants of \MinSum and \MinMax, where the electoral authority can choose any arborescence minimising the respective quantity and then apply an aggregation function. But what if they want to favour a specific alternative $d \in D = \{0, 1\}$? Without loss of generality, assume they want to favour alternative 1. Recall that for any arborescence $T,$ the agents $a \in N$ that vote 1 in the concrete votes are the descendants of 1 in $T.$ For both unravelling procedures, we will show that, surprisingly, if $N_1$ is the set of agents who vote for 1 in at least one optimal arborescence, then there is an optimal arborescence where all agents in $N_1$ vote for 1, and such an arborescence can be computed in polynomial time. This has immediate implications for the electoral authority attempting to exercise control: they can compute and select such an arborescence. The respective concrete votes subsume all possible concrete votes for 1 found in any optimal arborescence, so this leads to the best possible outcome for the authority, no matter which monotonic aggregation function is applied to the concrete votes. For \MinMax, the proof is not difficult.

\begin{theorem}\label{th:favor1} 
Let $N_1 \subseteq N$ be the agents who vote 1 in some minimum bottleneck arborescence (i.e., \MinMax certificate). Then, we can compute in $O(n + m)$ time a minimum bottleneck arborescence where all agents in $N_1$ vote 1. 
\end{theorem}

\begin{proof} 
For convenience, we work on the transposed graph. Let $w^*$ be the maximum edge weight used in a minimum bottleneck arborescence. This can be computed in $O(n + m)$ time using 
\cref{thm:efficient-minsum-minmax}. Write $E_{\leq x} = \{e \in E \mid w(e) \leq x\}$ for the set of edges of weight at most $x$ and $G_{\leq x} = (V, E_{\leq x})$ for the graph restricted to such edges. Then, any arborescence of $G_{\leq w^*}$ is a minimum bottleneck arborescence, and vice-versa. For readability, assume we remove all other edges and set $G = G_{\leq w^*}.$ 
Consider the set $S$ of vertices reachable from 1 in $G.$ Note that in all arborescences vertices in $(V \setminus \{r\}) \setminus S$ will be descendants of 0, so, if we can find an arborescence where all vertices in $S$ are descendants of 1, then the conclusion follows. Such an arborescence can indeed easily be constructed by running a depth-first search starting from the root and recursing along the edge to 1 before the edge to 0. This takes $O(n + m)$ time. Note that we implicitly get that $N_1 = S \setminus \{1\}.$ 
\end{proof}

The proof for \MinSum, while similar in spirit, requires insights from Fulkerson's primal-dual algorithm for computing minimum cost arborescences \citep{fulkerson}. Modern accounts of the algorithm are available in \cite{kamiyama_survey,utke2023anonymous,blocking_arborescences}. The algorithm outputs a collection $E' \subseteq E$ of edges of the graph, called the ``tight'' edges, as well as a \emph{laminar family} $\calL$ of subsets of $V \setminus \{r\}.$ Laminarity means that for any $L_1, L_2 \in \calL$ either $L_1 \subseteq L_2,$ or $L_2 \subseteq L_1$ or $L_1 \cap L_2 = \varnothing.$ Laminar families can be understood through their associated forest, where each set $L \in \calL$ has a set of children sets $\mathit{ch}_\calL(L)$ consisting of those $L' \in \calL$ such that $L' \subsetneq L$ and there is no $L'' \in \calL$ satisfying $L' \subsetneq L'' \subsetneq L.$ The set of \emph{roots} of the forest is denoted $\mathit{rt}_\calL$ and consists of all maximal sets in $\calL.$ Fulkerson's algorithm is given in the appendix. For this section, it suffices to understand a few key properties of its output known from previous works. Most importantly, the relationship between minimum cost arborescences and $(E', \calL)$ is given by the following powerful lemma, for which given a set of vertices $S \subseteq V,$ we write $\rho(S) = \{(u, v) \in E \mid u \in S\text{ and }v \notin S\}.$

\begin{restatable}{lemma}{lemmafulkersoncharact}\label{lemma:fulkerson-charact} An arborescence $T$ has minimum cost iff the following two conditions hold:~(i) $T \subseteq E';$ (ii) for any $L \in \calL$ it holds that $|\rho(L) \cap T| = 1.$
\end{restatable}

This gives a combinatorial characterisation of minimum-cost arborescences without mentioning edge weights, enabling us to adapt our approach for \MinMax to \MinSum.

\begin{restatable}{theorem}{thmmanipulationminsum}\label{thm:manipulation-minsum} Let $N_1 \subseteq N$ be the set of agents who vote 1 in some min-cost arborescence (i.e., \MinSum certificate). Then we can compute in polynomial time a min-cost arborescence where all agents in $N_1$ vote 1. Moreover, $a \in N_1$ iff 1 is reachable from $a$ along edges in $E'$ output by Fulkerson's algorithm.
\end{restatable}

On the positive side, Theorems~\ref{th:favor1} and ~\ref{thm:manipulation-minsum} suggest biased tie-breaking rules for \MinMax and \MinSum, which we denote by $\MinMax^1$ and $\MinSum^1$, selecting the collection of concrete votes $(X_a)_{a \in N}$ with the highest possible number of ones. Our results show that these procedures are resolute and subsume all possible votes for 1 among \MinMax (\MinSum) certificates. One can also define $\MinMax^0$ and $\MinSum^0$ analogously. These biased procedures can find use in settings where the decision to be made is for changing a status quo, where one can select between biasing for or against the status quo prior to the vote. We note the following attractive corollary of our results, for which given two collections of concrete votes $\vec{X}, \vec{Y}$ we write $\vec{X} \leq \vec{Y}$ to mean that $X_a \leq Y_a$ for all $a \in N$:

\begin{corollary}\label{coro:min-max-elements}
For any ballot list $\vec{B}$ and $\vec{X} \in \MinMax(\vec{B})$, 
$\vec{X'} \in \MinSum(\vec{B})$,
we have $\MinMax^0(\vec{B}) \leq \vec{X} \leq \MinMax^1(\vec{B})$ and
$\MinSum^0(\vec{B}) \leq \vec{X'} \leq \MinSum^1(\vec{B})$.
\end{corollary}

\section{Axiomatic Results}

In this section we study \MinSum and \MinMax from a normative standpoint. We introduce a natural axiom, which we call \emph{cast monotonicity}, and establish that it is satisfied by \MinSum and a lexicographic refinement of \MinMax (but not \MinMax itself). This holds both for the irresolute procedures and for their biased resolute variants. Unless stated otherwise, in this section we consider the classic model.

Fix a set of agents $N$ and $D = \{0, 1\}.$ Let $F$ be a (possibly irresolute) unravelling procedure mapping collections of ballots $\vec{B} = (B_a)_{a \in N}$ to collections of concrete votes $\vec{X} = (X_a)_{a \in N}.$ 
Given an 
aggregation function $\agg: \{0, 1\}^N \to \{0, 1\}$, the set of possible winning alternatives with respect to $F$ on input $\vec{B}$ is $\agg(F(\vec{B})) = \{\agg(\vec{X}) \mid \vec{X} \in F(\vec{B})\}.$ 
\begin{definition}
An unravelling procedure $F$ is \emph{cast-monotonic} with respect to a monotonic aggregation function $\agg$ if for every ballot list $\vec{B},$ every agent $a \in N,$ each alternative $d \in D$ and the ballot list $\vec{B'}$ obtained from $\vec{B}$ by changing $B_a$ to a direct vote for $d$ we have: 
(i) $d \in \agg(F(\vec{B}))$ implies $d \in \agg(F(\vec{B'}))$ and 
(ii) $1-d \notin \agg(F(\vec{B}))$ implies $1-d \notin \agg(F(\vec{B'})).$ 
We say that $F$ is \emph{(unconditionally) cast-monotonic} if it is cast-monotonic with respect to all monotonic aggregation functions $\agg.$
\end{definition}
Intuitively, cast monotonicity says that if an agent likes an alternative $d,$ then she is best off voting for it directly:~changing from an arbitrary ballot to a direct vote for $d$ 
can neither remove $d$ from the set of winners nor add $1-d$ to the set of winners.

Given two lists of concrete votes $\vec{X}, \vec{X'},$ recall that $\vec{X} \leq \vec{X'}$ if $X_a \leq X'_a$ for all $a \in N.$ We write $\leq_d$ to mean $\leq$ for $d = 1$ and $\geq$ for $d = 0.$ 
Then, we can give the following characterisation of unconditional cast monotonicity.

\begin{restatable}{lemma}{lemmaequivalpointwiseincrease}\label{lemma:equival-pointwise-increase} 
An unravelling procedure $F$ is cast-monotonic iff for every ballot list $\vec{B},$ every agent $a \in N,$ every alternative $d \in D,$ and the ballot list $\vec{B'}$ obtained from $\vec{B}$ by changing $B_a$ to a direct vote for $d,$ the following conditions hold: (iii) for every $\vec{X} \in F(\vec{B})$ there is an $\vec{X'} \in F(\vec{B'})$ with $\vec{X} \leq_d \vec{X'};$ (iv) for every $\vec{X'} \in F(\vec{B'})$ there is an $\vec{X} \in F(\vec{B})$ with $\vec{X} \leq_d \vec{X'}.$ Note that if $F$ is resolute, (iii) is equivalent to (iv).
\end{restatable}

It is not hard to show that \MinMax and its biased variants are not cast-monotonic.

\begin{restatable}{theorem}{thcastmax}\label{th:cast-max}
    $\MinMax$ fails cast monotonicity. The same holds for its resolute variants $\MinMax^p$ with $p \in \{0, 1\}.$
\end{restatable}

Next, we show that \MinSum and its variants are cast-monotonic. Our proof (see the appendix) relies on somewhat technical arguments regarding Fulkerson's algorithm.

\begin{restatable}{theorem}{thmminsumiscastmonotonic}\label{thm:min-sum-is-cast-monotonic} 
\MinSum, $\MinSum^0$, $\MinSum^1$ satisfy cast monotonicity.
\end{restatable}

The reader may wonder whether cast monotonicity results for \MinSum extend to the more general setting where agents can delegate to monotonic Boolean functions. Sadly, this is not the case, even if we only allow binary $\lor$ or $\land$.

\begin{restatable}{theorem}{thcastsumfailsandor}\label{th:cast-sum-fails-and-or} $\MinSum_\mathcal{F}$ fails cast monotonicity as soon as either $\textsc{Or}_2 \subseteq \mathcal{F}$ or $\textsc{And}_2 \subseteq \mathcal{F}.$ The same holds for the resolute variants $\MinSum^p_\mathcal{F}$ for $p \in \{0, 1\}.$
\end{restatable}

To end the section, we observe that the lexicographic refinement of \MinMax satisfies cast monotonicity. In this refinement, we seek a certificate that minimises the number of agents $a \in N$ with $c_a = n - 1,$ breaking ties by the number of agents with $c_a = n - 2,$ and so on. Formally:
\begin{definition}[\LexiMin]
    Let $\vec{c}_1$ and $\vec{c}_2$ be two certificates. We say that $\vec{c}_1 \leq_{\mathit{lex}} \vec{c}_2$ if, after sorting the entries in $\vec{c}_1$ and $\vec{c}_2$ in non-increasing order, $\vec{c}_1$ is lexicographically no larger than $\vec{c}_2$.
    The unravelling procedure \textsc{LexiMin} returns a minimum consistent certificate with respect to $\leq_{\mathit{lex}}$.
\end{definition}

\LexiMin can be viewed as a variant of \MinSum where the edge weights $0, 1, 2, \ldots, n - 1$ in the delegation graph are replaced with weights $1, X, X^2, \ldots, X^{n - 1}$ for a large enough $X$ (e.g., $X = n + 2$ suffices). Indeed, our results for \MinSum do not require the edge weights to be drawn from  $\{0, \dots, n - 1\}$; rather, they hold whenever the weights form a strictly increasing sequence. Accordingly, the two biased versions of \LexiMin are resolute, easy to compute, and satisfy cast monotonicity. For more expressive function classes, \LexiMin inherits its intractability from that of \MinSum, as it coincides with it on ballots of maximum size 2. 

\section{Summary and Future Work}
We have extended the work of Colley et al.~(\citeyear{grandi}), by strengthening their complexity-theoretic results to a complete dichotomy for monotone Boolean functions. For the standard model, we provided near-optimal algorithms for the two unravelling procedures and efficient algorithms for computing an outcome where a given alternative wins, leading to attractive tie-breaking rules for the setup where one alternative is the status quo. We introduced the cast monotonicity axiom, which \MinSum and the lexicographic refinement of \MinMax satisfy, while traditional \MinMax does not; these results extend to the biased variants of the rules. 
 
One interesting direction for future work is to extend the model by allowing agents to assign cardinal values to their preferences, instead of just ordering them.
The cardinal voting model can be seen as a meaningful generalisation of the preferential voting model we have explored. Moreover, our results show that computing optimal unravellings is NP-hard for most reasonable function classes, but there is a case relevant to the real world that is left uncovered. Namely, one implementation 
requires agents to choose before the election if they want to be delegates or not.
Delegates have to disclose their votes while non-delegates can keep their vote secret.
Suppose we allow non-delegates to vote for a single Boolean function of delegates, but delegates can only delegate directly to other delegates. Then, hardness of unravelling no longer arises as cycles are contained in the pure fragment of the model, increasing expressive power while maintaining tractability.
An interesting question in this setting is under what conditions an agent prefers to join the delegates as opposed to remaining a voter.

\section*{Acknowledgments}
We would like to thank the anonymous reviewers for their constructive feedback and useful suggestions contributing to improving this paper. This work was supported by the AI Programme of The Alan Turing Institute. A preliminary version of this work has been presented at the Games, Agents, and Incentives Workshop of AAMAS'23, based on the first author's Master’s thesis, which would not have been possible without the third author’s supervision. We thank the workshop attendees for their interesting questions and remarks. The first author would additionally like to thank his good friends Victoria Walker, Dimitrios Iatrakis and Radostin Chonev for their assistance in this work.

{\fontsize{9.8pt}{10.8pt} \selectfont
\bibliography{aaai24}}

\clearpage
\appendix

\section{Technical Appendix to ``Unravelling Expressive Delegations: Complexity and Normative Analysis''}
\vspace{35pt}
\begin{abstract}
    This appendix contains the proofs omitted from the main body as well as additional supporting material. We moreover include a short section on axiomatic dichotomies and an observation concerning a hardness proof from previous work.
\end{abstract}

\subsection{Preliminaries}

\thfixedpoint*
\begin{proof} Consider the iterative procedure checking whether $\vec{c}$ is consistent. If it finishes with a complete collection of concrete votes $\vec{X},$ then this is a solution to the system, and it is the unique one given that all assignments were ``forced.'' Otherwise, assume the procedure terminates with an incomplete collection $\vec{X}$ where entries in $N_\bot = \{a \in N \mid X_a = \bot\}$ are undetermined. Then, create two completions of $\vec{X},$ one where the entries for $N_\bot$ are assigned to 0 and one where they are assigned to 1, and call them $\vec{X^0}$ and $\vec{X^1}$ respectively. Note that $\vec{X^0} \leq \vec{X^1}$ holds, which is defined to mean that $X^0_a \leq X^1_a$ for all $a \in N.$ Hence, because delegation is only allowed to monotonic functions, $B_a(c_a)(\vec{X}^0) \leq B_a(c_a)(\vec{X}^1)$ holds. We will show that $\vec{X^0}$ and $\vec{X^1}$ are two (distinct) solutions to the system. To do so, it is enough to check that for any $a \in N_\bot$ we have $B_a(c_a)(\vec{X}^0) = 0$ and $B_a(c_a)(\vec{X}^1) = 1.$ This amounts to checking that $B_a(c_a)(\vec{X}^0) = B_a(c_a)(\vec{X}^1)$ would lead to a contradiction. This is indeed the case, as it would imply by monotonicity that for any completion $\vec{X'}$ of $\vec{X}$ the value $B_a(c_a)(\vec{X'})$ only depends on $\vec{X},$ contradicting that $\vec{X}$ was the output of the algorithm.  
\end{proof}

\subsection{Complexity of $\MinSum_\calF$}

\thmorandtwohardtopolyapp*
\begin{proof}
    Follows from \cref{lemma:sumOrConstantFactorHardness,lemma:sumOrPolyHardness,lemma:sumAndPolyHardness} below.
\end{proof}

\begin{lemma}\label{lemma:sumOrConstantFactorHardness}
$\MinSum_{\textsc{Or}_2}$ is NP-hard and hard to constant-factor approximate, even for maximum ballot size 2.
\end{lemma}
    
\begin{proof}
    We reduce from the decision version of \textsc{VertexCover}. For consistency with some of our later proofs, we work with an equivalent formulation of Vertex Cover based on 2SAT. Namely, the input is a Boolean formula $\Phi = \bigwedge_{i = 1}^k C_i $ over variables $x_1, \ldots, x_n,$ where each clause $C_i$ is a disjunction of two positive literals, as well as a number $K$. We will construct an instance of $\MinSum_{\textsc{Or}_2}$ for which a consistent certificate of cost at most $K$ exists iff $\Phi$ has a satisfying assignment with at most $K$ variables set to 1.
    
    For each variable $x_i$ introduce a voter $x_i$ with voting profile $B_{x_i} = 0 \succ 1$. Note that this is not a valid smart ballot, but can be simulated by introducing a constant voter \underline{zero} with $B_{\text{\underline{zero}}} = 0$ and using $B_{x_i} = \text{\underline{zero}} \succ 1$ instead. Intuitively, variable $x_i$ will evaluate to $1$ precisely when voter $x_i$ gets preference level $1$.
    
    For each clause $C = x_i \lor x_j$ in $\Phi$ we construct $K + 1$ identical gadgets with the intuitive property that they incur no additional cost if $C$ is satisfied and a cost of at least $K + 1$ otherwise. One gadget consists of two fresh voters $a$ and $b$ with smart profiles $B_a = x_j \lor b \succ 0$ and $B_b = x_i \lor a \succ 0.$ This is illustrated in \cref{fig:MinSumHardness2}.
\tikzset{circ/.style={circle,draw, solid, inner sep=0pt,minimum size=10mm}}
    \begin{figure}
        \centering
        \begin{tikzpicture}[edge from parent/.style={draw,-latex}, align=center,node distance=1.75cm]
            \node[circ] (xi){$x_i$};
            \node[circ, right of=xi] (a){$a$};
            \node[circ, below of=xi] (xj){$x_j$};
            \node[circ, right of=xj] (b){$b$};
            
            \node[draw, dashed, rounded rectangle, fit={(xi)(a)}] (elli1) {};
            \node[draw, dashed, rounded rectangle, fit={(b)(xj)}](elli2) {};
            
            \draw [->] (b) -- (elli1) node[midway,left]{$\lor$};
            \draw [->] (a.east) to [bend left=90] node[right]{$\lor$} (elli2.east) {};
        \end{tikzpicture}
        \caption{Gadget of clause $x_a \lor x_b$ for $\MinSum_{\textsc{Or}_2}.$}
        \label{fig:MinSumHardness2}
    \end{figure}

    We now prove the correctness of the construction. First, assume $\vec{y} = y_1, \ldots, y_n$ is a satisfying assignment of $\Phi$ with at most $K$ variables set to $1.$ Then, we construct a certificate $\vec{c}$ of cost at most $K,$ as follows. For voters $x_i,$ set $c_{x_i} = y_i,$ while for other voters $v$ set $c_v = 0.$ We now show that $\vec{c}$ is consistent. To do so, we need to show that the preferences of gadget voters can be unravelled into votes according to $\vec{c}$. Consider a clause $C = x_i \lor x_j.$ Because $\vec{y}$ is a satisfying assignment, one of $y_i$ or $y_j$ is $1$. By symmetry, it suffices to consider $y_i = 1.$ Then, the vote of $b$ can be computed as $x_i \lor a = y_i \lor a = 1 \lor a = 1.$ Similarly, the vote of $a$ can be computed as $x_j \lor b = x_j \lor 1 = 1.$ Hence, $\vec{c}$ is consistent.

    Conversely, assume $\Phi$ is not satisfiable and let $\vec{c}$ be a consistent certificate. We will show that the cost of $c$ is at least $K + 1.$ From certificate $\vec{c},$ construct an assignment $\vec{y} = y_1, \ldots, y_n$ by setting $y_i = c_{x_i}.$ Because $\vec{y}$ does not satisfy $\Phi$, consider a clause $C = x_i \lor x_j$ which is not satisfied by $\vec{y};$ i.e., $y_i = y_j = 0.$ We now show that all $K + 1$ gadgets corresponding to $C$ each incur a cost of at least one. Consider one such a gadget and assume for a contradiction that it incurs a cost of zero; i.e., its voters $a$ and $b$ both get their first preference, namely $a$ gets $x_j \lor b$ and $b$ gets $x_i \lor a.$ However, since $y_i = y_j = 0,$ these preferences resolve to $0 \lor b = b$ and $0 \lor a = a,$ which is a preference cycle, contradicting the consistency of $c.$

    We now apply gap amplification to our proof to also get hardness of approximation. Suppose that in the reduction above instead of $K + 1$ copies of each gadget we used $c(K + 1)$ copies for some $c \geq 1.$ Then, by the same argument as above, if $\Phi$ is satisfiable with at most $K$ variables set to $1,$ there is a consistent certificate of cost at most $K,$ and otherwise all consistent certificates incur a cost of at least $c(K + 1).$ A polynomial-time $c$-approximation algorithm for $\MinSum_{\textsc{Or}_2}$ could be used to produce a solution of cost a most $cK < c(K + 1)$ when $\Phi$ is satisfiable with at most $K$ variables set to $1$, and hence can not exist unless P = NP. 
\end{proof}

\begin{lemma}\label{lemma:sumOrPolyHardness}
    $\MinSum_{\textsc{Or}_2}$ is hard to $n^{1 - \epsilon}$-factor approximate even for maximum ballot size 2.
\end{lemma}
\begin{proof}
To get the stronger inapproximability result, we will set $c = n^k$ in our previous proof, with $k$ to be chosen later. Care needs to be taken here because, to talk about a $c$-approximation, $c$ needs to be expressed as a function of the $\MinSum_{\textsc{Or}_2}$ instance size.\footnote{So we do \emph{not} get that an $n^k$-approximation does not exist.} To avoid notational clash, in this proof we use $n$ to denote the number of variables in $\Phi$ and $X$ to denote the number of voters in instances of $\MinSum_{\textsc{Or}_2}.$
Assume for a contradiction that $\MinSum_{\textsc{Or}_2}$ admits a polynomial time $X^{1 - \epsilon}$-approximation algorithm $\mathcal{A}$ for some $\epsilon > 0.$ Assuming $\Phi$ consists of $m$ clauses, the instance produced by our reduction consists of $X = n + 1 + 2n^k(K + 1)m \leq n + 1 + 2n^k(n + 1)n^3 = O(n^{k + 4})$ voters, which for some $C > 0$ is at most $Cn^{k + 4}.$ Therefore, for our instance algorithm $\mathcal{A}$ approximates within a factor of $X^{1 - \epsilon} \leq C^{1 - \epsilon}n^{(k + 4)(1 - \epsilon)}.$ To get a contradiction, we want that this is at most $c = n^k;$ i.e., $C^{1 - \epsilon}n^{(k + 4)(1 - \epsilon)} \leq n^k,$ which is the same as $C^{1 - \epsilon} \leq n^{k - (k + 4)(1 - \epsilon)} = n^{\epsilon(k + 4) - 4}.$ By picking $k \geq 5 \epsilon^{-1} - 4,$ the latter will hold for $n$ large enough, yielding a polynomial-time algorithm for deciding Vertex Cover, subject to $n$ being large enough, which is not possible unless P = NP.
\end{proof}

\begin{lemma}\label{lemma:sumAndPolyHardness}
    $\MinSum_{\textsc{And}_2}$ is hard to $n^{1 - \epsilon}$-factor approximate even for maximum ballot size 2.
\end{lemma}
\begin{proof}
    We only outline the changes required compared to the proofs for $\lor.$ Namely, variable voters now have $B_{x_i} = \text{\underline{one}} \succ 0,$ where \underline{one} is a constant voter with $B_{\text{\underline{one}}} = 1.$ Moreover, the clause gadgets now consists of two voters $a$ and $b$ with smart ballots $B_a = x_j \land b \succ 0$ and $B_b = x_i \land a \succ 0.$ The rest of the proof goes through as written, except for when it comes to calculating votes, where $\land$ replaces $\lor,$ and caution needs to be observed when substituting in values for $x_i$ because of the now ``inverted'' smart ballots used for the variable voters. For instance, where before we had calculated the vote of $b$ in the first implication as $x_i \lor a = y_i \lor a = 1 \lor a = 1,$ we now do it as $x_i \land a = B_b(c_{x_i}) \land a = B_b(y_i) \land a = (1 - y_i) \land a = 0 \land a = 0.$
\end{proof}

\subsection{Complexity of $\MinMax_\calF$}

\thmminmaxbooleanpolytime*

\begin{proof} We show this for $\MinMax_\textsc{Or},$ the case of $\MinMax_\textsc{And}$ is completely symmetric. We binary search for the smallest $w$ such that a consistent certificate $\vec{c}$ with $\max_{a \in N}c_a \leq w$ exists. To check for a given $w$ whether such a certificate exists and compute one if so, let us restrict the ballots to the first $w$ preference levels, namely performing $k_a \gets \min\{k_a, w\}.$ Now, given the new ballots $\vec{B}$ we want to find any consistent certificate (which might not be possible if we deleted too many backup votes).

The key observation is that each certificate $\vec{c}$ corresponds to a directed graph with vertex set $N$ and nodes labeled with elements from the set $\{0, 1, \bot\},$ constructed as follows:~label all nodes with $\bot,$ except those $a \in N$ such that $B_a(c_a) = \tau_a,$ which are labelled with $\tau_a \in \{0, 1\}.$ For all other nodes, $B_a(c_a) = b_1 \lor b_2 \lor \ldots \lor b_k$ holds for some $k \geq 1.$ Add to the graph all the edges in the set $E_a(c_a) = \{(a, b_i) \mid 1 \leq i \leq k\}.$ Then, running the iterative algorithm computing the concrete votes $\vec{X}$ for $\vec{c}$ will lead to the following:~nodes $a \in N$ that can reach a node labelled with 1 will satisfy $X_a = 1;$ from the other nodes, those part of no cycle will satisfy $X_a = 0$ and the other ones $X_a = \bot.$ Hence, a certificate is consistent if and only if nodes that can not reach a node labelled 1 in the corresponding graph form an acyclic induced subgraph.

Consequently, to find a consistent certificate, we should first try to maximize the set of nodes reaching a node labelled 1. To do this, we begin by assigning the backup vote to all agents still having it, if it happens to be a 1 (hence maximizing the set of nodes labelled 1). Then, consider another graph with vertex set $N,$ the nodes already assigned labelled 1 and the rest labelled $\bot,$ and with the edge set containing $E_a(i)$ for all $a$ labelled $\bot$ and $0 \leq i \leq k_a.$ Find the nodes that can reach a node labelled 1 in this graph and compute a directed spanning forest $T$ of the subgraph induced by these nodes with sinks the nodes labelled 1 (this can be done by running a depth-first search on the transposed graph from all sinks and returning the tree edges). For each non-sink node $a$ in this subgraph, let $(a, x) \in T$ be the corresponding
parent edge in $T.$ Subsequently, find some $i$ for which $(a, x) \in E_a(i)$ and set $c_a = i.$ We have now ensured that the maximum set of nodes can reach a node labelled 1 in the certificate. These nodes $N_1 \subseteq N$ are precisely the nodes $a \in N$ satisfying $c_a \neq \bot$\footnote{We took the liberty in this proof to use $c_a = \bot$ to denote an unset value for $c_a.$} thus far. For the remaining nodes $N \setminus N_1$, we can only hope to make their induced subgraph in the certificate acyclic. To do this, assign the backup vote to agents in $N \setminus N_1$ still having it (which has to be a 0). Then, proceed iteratively:~whenever an agent with $c_a = \bot$ exists such that for some $0 \leq i \leq k_a$ all edges $(a, b) \in E_a(i)$ satisfy $c_b \neq \bot,$ assign $c_a = i.$ If this process exhausts all of $N \setminus N_1,$ then we have found a consistent certificate, otherwise none exist. Intuitively, this iterative process computes a topological sorting of the subgraph induced by $N \setminus N_1$ in the certificate, one agent at a time.
\end{proof}

\thhardnessminmax*

\begin{proof} Follows from \cref{lemma:hardness-min-max,lemma:inapprox-min-max} below.
\end{proof}

\begin{lemma}\label{lemma:hardness-min-max} $\MinMax_{\textsc{OrAnd}_2}$ is NP-hard even for maximum ballot size 3.
\end{lemma}
\begin{proof}
    We reduce from the version of 3SAT where all clauses consist of either three positive or three negative literals, which is NP-hard by Schaefer's Dichotomy Theorem. 
    The input is a Boolean formula $\Phi = \bigwedge_{i = 1}^k C_i $ over variables $x_1, \ldots, x_n,$ where each clause $C_i$ consists of either three positive or three negative literals. We will construct an instance of $\MinMax_{\textsc{OrAnd}_2}$ for which a consistent certificate $\vec{c}$ with $\max{\vec{c}} \leq 1$ exists iff $\Phi$ is satisfiable.
    
    For each variable $x_i$ introduce a voter $x_i$ with voting profile $B_{x_i} = 0 \succ 1$. Note that this is not a valid smart ballot but can be simulated by introducing a constant voter \underline{zero} with $B_{\text{\underline{zero}}} = 0$ and using $B_{x_i} = \text{\underline{zero}} \succ 1$ instead. Intuitively, variable $x_i$ will evaluate to $1$ precisely when voter $x_i$ gets preference level $1$.
    
    For each clause $C$ in $\Phi$ we construct a gadget as follows:~if $C = x_i \lor x_j \lor x_k,$ introduce six fresh voters $a, a', b, b', c, c'$ with preference profiles as follows, illustrated in \cref{fig:MinMax}:
    \begin{align*}
        B_a &= x_k \lor c \succ a' \succ 0
        &B_{a'} &= x_c \lor c \succ a \succ 0 \\
        B_{b} &= x_i \lor a \succ b' \succ 0
        &B_{b'} &= x_i \lor a \succ b \succ 0 \\
        B_{c} &= x_j \lor b \succ c' \succ 0
        &B_{c'} &= x_j \lor b \succ c \succ 0
    \end{align*}
    If, instead, $C = \overline{x_i} \lor \overline{x_j} \lor \overline{x_k}$, then the gadget is the same but with $\lor$ exchanged for $\land.$ Intuitively, the clause gadget can be unravelled with only first and second preferences iff $C$ is satisfied by the assignment of the variable voters. This is made more rigorous later.

\tikzset{circ/.style={circle,draw, solid, inner sep=0pt,minimum size=10mm}}
    \begin{figure}
    \vspace*{-80pt}
    \begin{tikzpicture}[edge from parent/.style={draw,-latex}, align=center,node distance=1.75cm]
    
    \node[circ] (xi){$x_i$};
    \node[circ, right of=xi] (a){$a$};
    \node[circ, right of=a] (a'){$a'$};    

    \node[circ, below of=xi] (xj){$x_j$};
    \node[circ, right of=xj] (b){$b$};
    \node[circ, right of=b] (b'){$b'$};
    
    \node[circ, below of=xj] (xk){$x_k$};
    \node[circ, right of=xk] (c){$c$};
    \node[circ, right of=c] (c'){$c'$};

    \node[draw, dashed, rounded rectangle, fit={(xi)(a)}] (elli1) {};
    \node[draw, dashed, rounded rectangle, fit={(b)(xj)}](elli2) {};
    \node[draw, dashed, rounded rectangle, fit={(c)(xk)}](elli3) {};
    
    \draw [->, dashed] (a) to [bend right] (a');
    \draw [->, dashed] (a') to [bend right] (a);
    \draw[->] (a) .. controls (-3, 4) and (-2, -7) .. (elli3.south)  node[midway,left]{$\lor$};
    \draw[->] (a') .. controls (5, -2) and (5, -6) .. (elli3.south)  node[midway,right]{$\lor$};;

    \draw[->, dashed] (b) to [bend right] (b');
    \draw[->, dashed] (b') to [bend right] (b);
    \draw [->] (b) -- (elli1) node[midway,left]{$\lor$};
    \draw [->] (b') -- (elli1) node[midway,right]{$\lor$};

    \draw[->, dashed] (c) to [bend right] (c');
    \draw[->, dashed] (c') to [bend right] (c);
    \draw [-> ] (c) -- (elli2) node[midway,left]{$\lor$};
    \draw [->] (c') -- (elli2) node[midway,right]{$\lor$};
    \end{tikzpicture}
    \vspace*{-37pt}
    \caption{ Gadget of clause $x_i \lor x_j \lor x_k$ for $\MinMax_{\textsc{OrAnd}_2}.$}
    \label{fig:MinMax}
    \end{figure}

    We now prove the correctness of the construction. First, assume $\vec{y} = y_1, \ldots, y_n$ is a satisfying assignment of $\Phi.$ Then, we construct a consistent certificate $\vec{c}$ with $\max{\vec{c}} \leq 1,$ as follows. For voters $x_i,$ set $c_{x_i} = y_i.$ Now, consider a clause $C = x_i \lor x_j \lor x_k$ (the negative literals case is symmetric). Because $\vec{y}$ is a satisfying assignment, one of $y_i, y_j, y_k$ is $1$. By symmetry, it suffices to consider $y_i = 1.$ Then, the votes of $b$ and $b'$ can be computed using their first preferences as $x_i \lor a = y_i \lor a = 1 \lor a = 1.$ Similarly, the votes of $c$ and $c'$ can be computed using their first preferences as $x_j \lor b = x_j \lor 1 = 1$. Finally, the votes of $a$ and $a'$ can be computed using their first preferences as $x_k \lor c = x_k \lor 1 = 1.$ The described reasoning has produced a consistent certificate $\vec{c}$ with $\max{\vec{c}} \leq 1.$

    Conversely, assume $\vec{c}$ is a consistent certificate with $\max{\vec{c}} \leq 1.$ From certificate $\vec{c}$ construct an assignment $\vec{y} = y_1, \ldots, y_n$ by setting $y_i = c_{x_i}.$ We will show that $\vec{y}$ satisfies $\Phi.$ Assume for a contradiction some clause $C = x_i \lor x_j \lor x_k$ (the negative literals case is symmetric) is not satisfied by $\vec{y}$; i.e., $y_i = y_j = y_k = 0.$ Then, the preferences of the voters in the corresponding gadget resolve to the following by substituting $x_i, x_j, x_k$ with $0$:
    \begin{align*}
        B_a &= c \succ a' \succ 0
        &B_{a'} &= c \succ a \succ 0 \\
        B_{b} &= a \succ b' \succ 0
        &B_{b'} &= a \succ b \succ 0 \\
        B_{c} &= b \succ c' \succ 0
        &B_{c'} &= b \succ c \succ 0
    \end{align*}

    Clearly, unless there exists $v \in \{a, a', b, b', c, c'\}$ such that $c_v = 2,$ trying to unravel the preferences according to $\vec{c}$ will lead to a preference cycle, making $\vec{c}$ not consistent. However, this means that $\max{\vec{c}} \geq 2,$ a contradiction.
\end{proof}

\begin{lemma}\label{lemma:inapprox-min-max} $\MinMax_{\textsc{OrAnd}_2}$ is hard to constant-factor approximate.
\end{lemma}
\begin{proof}
This can be proven by modifying the proof of \cref{lemma:hardness-min-max}. We will still introduce three voters $a, b, c$ per clause. However, instead of voters $a', b', c'$ we will introduce $k$ voters of each kind: $(a'_\ell)_{\ell = 1}^k, (b'_\ell)_{\ell = 1}^k, (c'_\ell)_{\ell = 1}^k.$ We set the ballots of $a, b, c$ as follows:
\begin{gather*}
B_a = x_k \lor c \succ a'_1 \succ \ldots \succ a'_k \succ 0 \\
B_b = x_i \lor a \succ b'_1 \succ \ldots \succ b'_k \succ 0 \\
B_c = x_j \lor b \succ c'_1 \succ \ldots \succ c'_k \succ 0
\end{gather*}
\noindent We set the ballots of $a'_\ell, b'_\ell, c'_\ell$ as follows for $1 \leq \ell \leq k$:
{\small\begin{gather*}
B_{a'_\ell} = x_k \lor c \succ a'_1 \succ \ldots \succ a'_{\ell - 1} \succ a \succ a'_{\ell + 1} \succ \ldots \succ a'_k \succ 0 \\
B_{b'_\ell} = x_i \lor a \succ b'_1 \succ \ldots \succ b'_{\ell - 1} \succ b \succ b'_{\ell + 1} \succ \ldots \succ b'_k \succ 0 \\
B_{c'_\ell} = x_j \lor b \succ c'_1 \succ \ldots \succ c'_{\ell - 1} \succ c \succ c'_{\ell + 1} \succ \ldots \succ c'_k \succ 0
\end{gather*}}
For this instance, similar reasoning to the proof of \cref{th:hardness-min-max} shows that, if $\Phi$ is satisfiable, there is a consistent certificate with $\max \vec{c} \leq 1,$ and otherwise all consistent certificates satisfy $\max \vec{c} = k + 1.$ A polynomial-time $\alpha$-approximation algorithm for $\MinMax_{\textsc{OrAnd}_2}$ could be used to produce a consistent certificate $\vec{c}$ with $\max \vec{c} \leq \alpha$ (which is less than $k + 1$ if $k$ is chosen large enough) whenever $\Phi$ is satisfiable, and hence can not exist unless P = NP.
\end{proof}

\subsection{Computational Dichotomies}

Given a Boolean function $f \colon \{0,1\}^k \to \{0,1\}$ and a partial assignment of values to its variables $\nu : \{0, 1\}^k \to \{0, 1, \bot\}$ we write $f(\nu)$ for the function obtained from $f$ by substituting in the values where $\nu_i \neq \bot.$ Then, we have the following lemma.

\begin{lemma}\label{lemma:o-donnell} Let $f \colon \{0,1\}^k \to \{0,1\}$ be a monotonic function. Denote its arguments by $x_1, \ldots, x_k$ and write $K = \{1, \ldots, k\}.$ Then, unless $f = \bigvee_{i \in \calI}x_i$ for some $\calI \subseteq K,$ there exist indices $p \neq q$ and a partial valuation $\nu$ assigning values to all arguments except $x_p$ and $x_q$ such that $f(\nu) = x_p \land x_q.$ The same holds changing $(\bigvee, \land)$ to $(\bigwedge, \lor).$
\end{lemma}
\begin{proof} We prove the first part, as the second follows by applying the first to $f' = \overline{f(\overline{x_1}, \ldots, \overline{x_k})}.$ Write $f$ as a minimal monotone formula in DNF. Namely, let $C_1, \ldots, C_\ell \subseteq K$ be such that $f = \bigvee_{i = 1}^\ell\bigwedge_{j \in C_i}x_j.$ By minimality, for no $i \neq j$ does it hold that $C_i \subseteq C_j.$ If $\ell = 0$ or $|C_i| = 1$ for all $i$ then we are in the ``unless'' case. Otherwise, without loss of generality assume $|C_1| > 1$ and take $p \neq q$ such that $p, q \in C_1.$ We construct $\nu$ by setting setting all variables for $K \setminus C_1$ to 0 and all variables for $C_1 \setminus \{p, q\}$ to 1. We now show $f(\nu) = x_p \land x_q.$ Note that $\bigwedge_{j \in C_i}x_j$ evaluates to 0 under $\nu$ for all $i > 1$ because $C_i \nsubseteq C_1,$ from which there exists $b \in C_i \setminus C_1 \subseteq K \setminus C_1$ implying $\nu_b = 0.$ It remains to see that $\bigwedge_{j \in C_1}x_j = x_p \land x_q$ under $\nu,$ which is immediate given that $\nu$ sets the other variables for $C_1$ to 1.
\end{proof}

As an example, $\textsc{Maj}(x, y, 0) = x \land y$ and $\textsc{Maj}(x, y, 1) = x \lor y.$ Note that the two ``unless'' cases above correspond to $f \in \textsc{Or}$ and $f \in \textsc{And}.$ We are now ready to prove our full dichotomies.

\thmdich*

\begin{proof} Follows from \cref{lemma:dich-minsum,lemma:dich-minmax} below.
\end{proof}

\begin{lemma}\label{lemma:dich-minsum} Assume $\calF \subseteq \textsc{Mon}.$ If $\calF = \LIQUID,$ $\MinSum_\calF$ is poly-time computable. Otherwise, 
$\MinSum_\calF$ is NP-hard and hard to $n^{1 - \epsilon}$-approximate even for ballots of size at most 2.
\end{lemma}
\begin{proof} The easiness part follows from previous work (and also from \cref{thm:efficient-minsum-minmax}). For the hardness part, assume $\LIQUID \subsetneq \calF$ and consider some $f \in \calF \setminus \LIQUID.$ Notice that $\LIQUID = \textsc{Or} \cap \textsc{And},$ so $f \in (\calF \setminus \textsc{Or}) \cup (\calF \setminus \textsc{And}).$ First, consider the case $f \in \calF \setminus \textsc{Or}.$ By \cref{lemma:o-donnell}, let $\nu$ be a partial valuation assigning values to all variables of $f$ except say $x_p$ and $x_q$ such that $f(\nu) = x_p \land x_q.$ Assuming we introduce constant voters \underline{zero} and \underline{one} with $B_{\text{\underline{zero}}} = 0$ and $B_{\text{\underline{one}}} = 1$ if not already present, we can now simulate the proof of \cref{lemma:sumAndPolyHardness} replacing 0/1 in $\nu$ by \underline{zero}/\underline{one} and $\land$ by $f(\nu).$ The case $f \in \calF \setminus \textsc{And}$ is analogous by using the other half of \cref{lemma:o-donnell}, and 
\cref{lemma:sumOrConstantFactorHardness,lemma:sumOrPolyHardness} instead of \cref{lemma:sumAndPolyHardness}. 
\end{proof}

\begin{lemma}\label{lemma:dich-minmax} Assume $\calF \subseteq \textsc{Mon}.$ If $\calF \subseteq \textsc{Or}$ or $\calF \subseteq \textsc{And},$ $\MinMax_\calF$ is poly-time computable. Otherwise, $\MinMax_\calF$ is NP-hard even for ballots of size at most 3, and hard to constant-factor approximate.
\end{lemma}
\begin{proof} The easiness part is \cref{thm:minmax-Boolean-poly-time}. For the hardness part, there exist $f \in \calF \setminus \textsc{Or}$ and $f' \in \calF \setminus \textsc{And}.$ Applying \cref{lemma:o-donnell} to $f$ and $f'$ we can simulate binary $\land$ and $\lor$ in the proof of \cref{th:hardness-min-max} similarly to the proof of \cref{lemma:dich-minsum}.
\end{proof}

\subsection{Efficient Computation of \MinSum and \MinMax Certificates}

\thmefficientminsumminmax*
\begin{proof} Follows from \cref{lemma:minsum-efficient,lemma:minmax-efficient} below.
\end{proof}

\begin{lemma}\label{lemma:minsum-efficient} A \MinSum certificate can be computed in time $O(n \log n + m).$
\end{lemma}

\begin{proof} Direct consequence of applying the algorithm of \citet{gabow_min_cost_arb} to compute a minimum cost arborescence. This is essentially a clever implementation of the Chu–Liu-Edmonds algorithm, the latter running in time $O(nm).$ 
\end{proof}

We leave open whether the running time can be improved using the fact that edge weights are small non-negative integers (i.e., less than $n$). Note that for instances with average ballot size $\Omega(\log n),$ this is already asymptotically optimal. 

\begin{lemma}\label{lemma:minmax-efficient} A \MinMax certificate can be computed in time $O(n + m)$.
\end{lemma}
\begin{proof}
For convenience, we work on the transposed graph. Write $E_x = \{e \in E \mid w(e) = x\}$ for the set of edges of weight $x$ and $E_{\leq x} = \cup_{x' = 0}^{x}E_{x'}.$ An arborescence using edges with weights at most $x$ exists iff all vertices can reach the root $r$ in $G_{\leq x} = (V, E_{\leq x});$ i.e., $G_{\leq x}$ admits an arborescence. 
As edge weights are between $0$ and $n - 1,$ this already gives an $O((n + m) \log n)$ algorithm:~perform binary search for the smallest $0 \leq x < n$ such that $G_{\leq x}$ admits an arborescence and report some arborescence of $G_{\leq x}.$

To improve on this, compute the sets $(E_x)_{0 \leq x < n}$ and then proceed starting with $x = 0$ and increase $x$ by one until all vertices in $G_{\leq x}$ can reach the root. To do this efficiently, maintain a set $S$ of nodes that can reach the root at any point in time. Initially, $S = \{r\},$ and we stop whenever $S = V.$ Moreover, maintain the current graph $G_{\leq x}$ using adjacency lists and add the new edges $E_{x + 1}$ to it one by one whenever $x$ is incremented. When adding an edge $(u, v),$ add it to $u$'s adjacency list and check whether $u \in S$ and $v \notin S.$ If so, launch a graph traversal (e.g., depth-first search) from $v$ restricted to vertices in $V \setminus S$ and edges already in the adjacency lists, adding to $S$ any newly discovered vertices. Because the graph traversals will enter each vertex exactly once and hence each edge $e \in E$ will be considered at most once in the traversals, the total time complexity is $O(n + m).$
\end{proof}

\subsection{Control for \MinSum}

\begin{algorithm}[t]
\caption{
Fulkerson's Algorithm}\label{alg:fulkerson}
\textbf{Input}: Directed graph $G = (V, E)$ rooted at $r \in V,$ weight function $w : E \to \mathbb{R}_{\geq 0}.$\\
\textbf{Output}: Set of tight edges $E' \subseteq E$ and laminar family $\calL$ of subsets of $V \setminus \{r\}.$
\begin{algorithmic}[1]
\State $E', \calL \gets \varnothing$
\While{there exists a strongly connected component  $L$ of $G' = (V, E')$ such that $r \notin L$ and $\rho(L) \cap E' = \varnothing$}
\State Pick one such $L$ arbitrarily. \label{line:nondet}
\State $\calL \gets \calL \cup \{L\}$
\State $w_\mathit{min} \gets \min_{e \in \rho(L)}w(e)$
\State $w(e) \gets w(e) - w_\mathit{min}$ for all $e \in \rho(L).$
\State $E' \gets E' \cup \{e \in \rho(L) \mid w(e) = 0\}$
\EndWhile
\State \textbf{return} $E'$ and $\calL.$
\end{algorithmic}
\end{algorithm}

Fulkerson's primal-dual algorithm for computing minimum cost arborescences \citep{fulkerson,kamiyama_survey,utke2023anonymous,blocking_arborescences} is given in \cref{alg:fulkerson}.\footnote{Technically, only the first phase of the algorithm. The second phase builds a min-cost arborescence similarly to \cref{lemma:recursive} below.} To prove our result for control of \MinSum, we will need a number of key properties of the output $(E', \calL)$ of the algorithm known from previous works. The most prominent one is \cref{lemma:fulkerson-charact}, stated in the main text, which we repeat here for convenience.

\lemmafulkersoncharact*

This gives a combinatorial characterization of minimum cost arborescences without mentioning edge weights. It can be shown that the output of the algorithm does not depend on the nondeterministic choice of $L$ made on \cref{line:nondet}, so it is meaningful and convenient to talk about ``the output'' of the algorithm, rather than ``an output,'' but our proofs can all in principle be written without using this fact. We will also need the following property, which follows by directly examining the algorithm.

\begin{lemma}\label{lemma:scc} For any $L \in \calL,$ the subgraph induced by $L$ is strongly connected with respect to tight edges.
\end{lemma}

We begin with the following auxiliary result.

\begin{lemma}\label{lemma:recursive} For any $L \in \calL$ and $v \in L,$ there exists a $v$-arborescence of the induced subgraph $G[L]$ using only tight edges and satisfying condition (ii) for all $L' \in \calL$ such that $L' \subsetneq L.$ Moreover, it can be computed in polynomial time.
\end{lemma}
\begin{proof} For clarity of exposition, remove all non-tight edges from $G$ and restrict our attention to $G[L].$ Moreover, also for clarity, and without affecting the correctness of any of our previous claims, assume that $\calL$ contains all singleton sets $\{v\}$ for $v \in V \setminus \{r\}.$

We construct the required arborescence recursively. First, enumerate $\mathit{ch}_\calL(L) = \{L_1, \ldots, L_k\}.$ Note that $L = \cup_{i = 1}^kL_i.$ Without loss of generality, assume $v \in L_1.$ For $1 \leq i \leq k$ contract $L_i$ to a single new node called $x_i,$ keeping track of the correspondence between new edges and old edges.\footnote{This can introduce multi-edges, we keep them and their costs.} This gives a new graph with vertex set $\{x_1, \ldots, x_k\}.$ By \cref{lemma:scc}, the subgraph induced by $L$ is strongly connected, and this is not altered by contractions, so the new graph is also strongly connected. As a result, a graph traversal started from $x_1$; e.g., depth-first search; can be used to find an arborescence rooted at $x_1$ in the contracted graph. This arborescence corresponds to edges $(a_i, b_i)_{i = 2}^k$ in the original graph $G[L],$ where $a_i \in L_i.$ We make these edges part of the answer arborescence. To get the remaining edges, recurse inside each $L_i$ to get an arborescence in $G[L_i]$ rooted at $a_i$ with the required properties, and add it to the answer arborescence.
\end{proof}

We can now show our result for control of \MinSum by building on the proof idea for \MinMax.

\thmmanipulationminsum*
\begin{proof} Run Fulkerson's algorithm, and let $(E', \calL)$ be the output. The minimum-cost arborescences are precisely those satisfying conditions (i) and (ii) in \cref{lemma:fulkerson-charact}. For readability, assume we remove all edges in $E \setminus E'$ from the graph. Then, we are interested in arborescences satisfying condition (ii) only. Consider the set $S$ of vertices that can reach 1 in $G.$ Note that, in all arborescences, vertices in $(V \setminus \{r\}) \setminus S$ will be descendants of 0, so, if we can find an arborescence satisfying (ii) where all vertices in $S$ are descendants of 1, then the conclusion follows. We next show this last fact.

First, note that for any $L \in \calL,$ either $L \subseteq S,$ or $L \subseteq (V \setminus \{r\}) \setminus S.$ To see this, note that if $u, v \in L$ are such that $u \in S$ and $v \in (V \setminus \{r\}) \setminus S,$ then, by \cref{lemma:scc}, there would be a path from $v$ to $u,$ implying that $v \in S,$ which can not be the case. To construct the arborescence, the approach is largely similar to the one taken in the proof of \cref{lemma:recursive}. Contract all sets in $\mathit{rt}_\calL$ into single nodes and let $x$ be the node corresponding to 1. Because any $L \in \mathit{rt}_\calL$ is strongly connected and satisfies either $L \subseteq S$ or $L \subseteq (V \setminus \{r\}) \setminus S,$ it follows that the nodes in the contracted graph that can reach $x$ are precisely those corresponding to $L$'s for which $L \subseteq S.$ Then, construct an arborescence in the contracted graph by using (this time purposefully) depth-first search on the transposed graph starting from $r$ and recursing to $x$ before the node corresponding to 0. Then, as before, expand back the contracted nodes, marking the edges in the original graph corresponding to edges in the arborescence to be included in the answer arborescence. Finally, recurse inside each $L \in \mathit{rt}_\calL$ with roots chosen as before using \cref{lemma:recursive} to complete the arborescence. It can be checked that the output is a min-cost arborescence where agents in $S \setminus \{1\}$ vote for 1, from which implicitly we also get that $N_1 = S \setminus \{1\}.$
\end{proof}

\subsection{Axiomatic Results}

We will need the following auxiliary lemma:

\begin{lemma}\label{lemma:construct-monot} Consider an unravelling procedure $F,$ two arbitrary collections of ballots $\vec{B}, \vec{B'}$ and some $d \in D.$ Assume that for some $\vec{X} \in F(\vec{B})$ and all $\vec{X'} \in F(\vec{B'})$ it does not hold that $\vec{X} \leq_d \vec{X'}.$ Then, there exists a monotonic aggregation function $\agg : D^N \to D$ such that $d \in \agg(F(\vec{B}))$ and $d \notin \agg(F(\vec{B'}).$
\end{lemma}
\begin{proof} We consider the case $d = 1,$ as the other case is fully analogous. Write $S = \{a \in N \mid X_a = 1\},$ then we define the required monotonic aggregation function $\agg = \bigwedge_{a \in S}a.$ Note that for any input $\vec{Y},$ it holds that $\agg(\vec{Y}) = 1$ iff $\vec{X} \leq \vec{Y}.$ Hence, it follows that $\agg(\vec{X}) = 1$ and $\agg(\vec{X'}) = 0$ for all $\vec{X'} \in F(\vec{B'}).$ As a result, $1 \in \agg(F(\vec{B}))$ and $1 \notin \agg(F(\vec{B'})),$ as desired. 
\end{proof}

\lemmaequivalpointwiseincrease*

\begin{proof} Assume $\vec{B}, a, d$ are such that (iii) and (iv) hold. We want to show that (i) and (ii) also hold for any monotonic aggregation function $\agg.$

First, for (i), assume $d \in \agg(F(\vec{B})).$ By definition, there exists $\vec{X} \in F(\vec{B})$ such that $\agg(\vec{X}) = d.$ By (iii), there exists $\vec{X'} \in F(\vec{B'})$ such that $\vec{X} \leq_d \vec{X'}.$ By monotonicity, $\agg(\vec{X}) \leq_d \agg(\vec{X'}),$ so $d \leq_d \agg(\vec{X'}).$
Hence, $\agg(\vec{X'}) = d,$ implying that $d \in \agg(F(\vec{B'})),$ as required.

Then, for (ii), write $d' = 1 - d$ and assume $d' \in \agg(F(\vec{B'})),$ then we want to show that $d' \in \agg(F(\vec{B})).$ Since $d' \in \agg(F(\vec{B'})),$ let $X' \in F(\vec{B'})$ be such that $d' = \agg(\vec{X'}).$ By (iv), there exists $\vec{X} \in F(\vec{B})$ such that $\vec{X} \leq_d \vec{X'}.$ By monotonicity, $\agg(\vec{X}) \leq_d \agg(\vec{X'}),$ so $\agg(\vec{X}) \leq_d d'.$ Since $d' = 1 - d$ this means that $\agg(\vec{X}) = d',$ so $d' \in \agg(F(\vec{B})),$ as required.

Conversely, now assume $\vec{B}, a, d$ are such that one of (iii) or (iv) fails. We want to show that either (i) or (ii) fails for some monotonic aggregation function $\agg.$

First, assume (iii) fails, then it follows that (i) also fails by directly applying \cref{lemma:construct-monot}.

Then, assume (iv) fails. Consider some $\vec{X'} \in F(\vec{B'})$ such that for all $\vec{X} \in F(\vec{B})$ it does not hold that $\vec{X} \leq_d \vec{X'}.$ Write $d' = 1 - d,$ and note that $\vec{X} \leq_d \vec{X'} \iff \vec{X'} \leq_{d'} \vec{X}.$ Hence, applying \cref{lemma:construct-monot} with $d$ replaced by $d'$ and with the roles of $(\vec{B}, \vec{X})$ and $(\vec{B'}, \vec{X'})$ reversed, we get the existence of a monotone aggregation function $\agg$ such that $d' \in \agg(F(\vec{B'}))$ and $d' \notin \agg(F(\vec{B})),$ failing (ii), as required.
\end{proof}

\thcastmax*

\begin{proof} Consider agents $N = \{a, a', \text{\underline{zero}}, u_1, \ldots, u_{n-3}\}$ casting ballots $\vec{B}$ as follows:~$B_a = a' \succ 1, B_{a'} = a \succ 1, B_\text{\underline{zero}} = 0, B_{u_i} = \text{\underline{zero}} \succ 1$ for all $1 \leq i \leq n - 3.$ A certificate $\vec{c}$ is consistent iff at least one of $c_{a} = 1$ and $c_{a'} = 1$ holds, and, since all ballots have size at most two, all consistent certificates are \MinMax certificates. As a result, one can check that $\vec{X} \in \MinMax(\vec{B})$ iff $X_{a} = X_{a'} = 1.$

Now, consider ballots $\vec{B'}$ where the only difference is $a$ switches their vote to $B'_a = 1.$ Here, the unique \MinMax certificate has everyone delegating to their first preferences, leading to the corresponding concrete votes $\vec{X'} = (1, 1, 0, \ldots)$ from which $\MinMax(\vec{B'}) = \{\vec{X'}\}.$

Since $\vec{X} = (1, 1, 1, \ldots) \in \MinMax(\vec{B})$ and $\vec{X} \nleq \vec{X'},$ by \cref{lemma:equival-pointwise-increase}, $\MinMax$ fails cast monotonicity. For instance, this is the case for any odd $n \geq 5$ and the majority aggregation function $\agg.$ By the same reasoning, since $\MinMax^1(\vec{B}) = \vec{X}$ and $\MinMax^1(\vec{B'}) = \vec{X'},$ $\MinMax^1$ also fails cast monotonicity. To show the same also holds for $\MinMax^0$ consider the same instance with the roles of 0 and 1 inverted. 
\end{proof}

To show that \MinSum and its variants satisfy cast monotonicity, we require the following technical lemma about Fulkerson's algorithm concerning delegation graphs differing only in edges exiting a single node; we find it of independent interest.

\begin{lemma}\label{lemma:parallel-fulkerson-executions}
For each vertex $a \in V \setminus \{r, 0, 1\}$ in the delegation graph there exists a set $E'_a \subseteq E$ such that, no matter how the edges exiting $a$ and their costs are altered, the set of tight edges $E'$ output by Fulkerson's algorithm satisfies: (i) $E_a' \subseteq E';$ (ii) for any $(x, y) \in E' \setminus E'_a$ vertex $x$ can reach $a$ but not $r$ along edges in $E'_a;$ (iii) if there is a single edge exiting $a$, namely $(a, d)$ for some $d \in D,$ then  $E' = E_a' \cup \{(a, d)\}.$
\end{lemma}
\begin{proof} 
Consider running Fulkerson's algorithm with the following modification:~change the loop condition to additionally require $a \notin L$ and never pick $L$ with $a \in L$ in \cref{line:nondet}. Let $(E'_a, \calL_a)$ be the respective output. Note that, irrespective of the edges exiting node $a$ and their weights, if $(E', \calL)$ is the output of the full algorithm, then $E'_a \subseteq E'$ and $\calL_a \subseteq \calL,$ implying (i). To see this, recall that the output of the full algorithm does not depend on the nondeterministic choices made on \cref{line:nondet}, so one can arrange for the full algorithm to be carried out by first executing the modified algorithm (which does not depend on edges exiting $a$) and the performing any supplementary steps.

Moreover, note that, with respect to edges in $E'_a,$ every vertex can reach either $r$ or $a$. To see this, assume that this was not the case for some vertex $x \in V \setminus \{r, a\}.$ Then, all vertices reachable from $x$ also fail the property. In particular, let $C$ be a sink strongly connected component; i.e., $\rho(C) \cap E'_a = \varnothing;$ reachable from $x,$ then it follows that $r, a \notin C.$ Hence, $C$ satisfies the conditions to execute one more iteration of the while loop of the modified algorithm with $L = C,$ contradicting the algorithm's termination. As a result, assuming the full algorithm is carried out by first executing the modified algorithm and then performing any necessary supplementary steps, the additional steps will always select $L$ satisfying $a \in L$ and $L \subseteq V_a,$ where we wrote $V_a \subseteq V \setminus \{r\}$ for the set of vertices that can reach $a$ but not $r$ along edges in $E'_a.$ To see why this is the case, note that for all additional steps, the only strongly connected component excluding $r$ and satisfying $\rho(L) \cap E' = \varnothing$ will be that of $a$ as all vertices can reach either $a$ or $r$ along tight edges. Moreover, the process terminates the first time such an $L$ is selected and an edge $(x, y) \in \rho(V_a)$ becomes tight, as then all vertices will be able to reach $r.$ This implies (ii), namely that all edges $(x, y)$ added during supplementary steps satisfy $x \in V_a.$

To see (iii) also holds, note that if $(a, d)$ is the only edge exiting $a,$ then, after executing the modified algorithm, the full algorithm will terminate in the next step by making this edge tight. This is because the strongly connected component of $a$ after terminating the modified algorithm consists solely of $a$ since no previous steps have made any edges exiting $a$ tight.
\end{proof}

\thmminsumiscastmonotonic*
\begin{proof} Follows from \cref{lemma:minsum-resolute-is-cast-monotonic,lemma:min-sum-is-cast-monotonic} below.
\end{proof}

\begin{lemma}\label{lemma:minsum-resolute-is-cast-monotonic} $\MinSum^0$ and $\MinSum^1$ are cast-monotonic.     
\end{lemma}
\begin{proof} Since both rules are resolute, we show this by verifying condition (iii) of  \cref{lemma:equival-pointwise-increase}. Consider any collection of ballots $\vec{B},$ any agent $a \in N$ and, without loss of generality, $d = 1.$ Write $\vec{B'}$ for $\vec{B}$ with $B_a$ changed to a direct vote for $1$. We want to show that $\MinSum^p(\vec{B}) \leq \MinSum^p(\vec{B'})$ for $p \in \{0, 1\}.$ Let $E_\vec{B}'$ and $E_\vec{B'}'$ be the tight edges output by Fulkerson's algorithm executed on the delegation graphs of $\vec{B}$ and $\vec{B'}$ respectively. Note that the two delegation graphs only differ in the edges exiting $a \in N = V \setminus \{r, 0, 1\},$ so apply \cref{lemma:parallel-fulkerson-executions} to get a set $E'_a \subseteq E_\vec{B}' \cap E_\vec{B'}'.$ Moreover, write $V_a \subseteq V \setminus \{r\}$ for the set of vertices that can reach $a$ but not $r$ along edges in $E'_a.$ Then any $(x, y) \in (E_\vec{B}' \cup E_\vec{B'}') \setminus E'_a$ satisfies $x \in V_a.$ Finally, since $(a, 1)$ is the only edge exiting $a$ for $\vec{B'},$ we have $E'_{\vec{B'}} = E'_a \cup \{(a, 1)\}.$

First, we show that $\MinSum^1(\vec{B}) \leq \MinSum^1(\vec{B'}).$ By the last part of \cref{thm:manipulation-minsum}, this amounts to showing that for any $x \in N,$ if 1 is reachable from $x$ along edges in $E_\vec{B}',$ then it is also reachable from $x$ along edges in $E'_{\vec{B'}} = E'_a \cup \{(a, 1)\}.$ Assume a path $P$ from $x$ to 1 along edges in $E_\vec{B}'$ exists. If $P$ only uses edges in $E'_a,$ then the conclusion is immediate. Otherwise, $P$ uses some edge in $E_\vec{B}' \setminus E'_a.$ Let $(i, j) \in E_\vec{B}' \setminus E'_a$ be the first such edge encountered along $P$ if traversed from $x$ to 1. This implies that $i \in V_a,$ so there is a path along edges in $E_a'$ from $i$ to $a.$ Combining the fragment of $P$ before edge $(i, j),$ the path from $i$ to $a$ along edges in $E_a'$ and the edge $(a, 1),$ we get a path from $x$ to $1$ using only edges in $E'_{\vec{B'}} = E'_a \cup \{(a, 1)\},$ as desired.

Next, we show that $\MinSum^0(\vec{B}) \leq \MinSum^0(\vec{B'}).$ By the last part of \cref{thm:manipulation-minsum} with the roles of 0 and 1 inverted, this amounts to showing that for any $x \in N,$ if 0 is reachable from $x$ along edges in $E_\vec{B'}' = E'_a \cup \{(a, 1)\},$ then it is also reachable from $x$ along edges in $E'_{\vec{B}}.$ Assume a path $P$ from $x$ to 0 along edges in $E_\vec{B'}' = E'_a \cup \{(a, 1)\}$ exists. By construction of the delegation graph, $P$ can not use the edge $(a, 1),$ so it really is a path along edges in $E'_a.$ Since $E'_a \subseteq E_\vec{B}',$ we get the conclusion.
\end{proof}

\begin{lemma}\label{lemma:min-sum-is-cast-monotonic} \MinSum satisfies cast monotonicity.
\end{lemma}
\begin{proof} In short, this follows by combining the results for the resolute variants in \cref{lemma:minsum-resolute-is-cast-monotonic} with \cref{coro:min-max-elements}.

We show this by verifying the conditions (iii) and (iv) of \cref{lemma:equival-pointwise-increase}. Consider any collection of ballots $\vec{B},$ any agent $a \in N$ and, without loss of generality, $d = 1.$ Write $\vec{B'}$ for $\vec{B}$ with $B_a$ changed to a direct vote for $1$. We want to show that:~(iii) for any $\vec{X} \in \MinSum(\vec{B})$ there is $\vec{X'} \in \MinSum(\vec{B'})$ such that $\vec{X} \leq \vec{X'};$ (iv) for any $\vec{X'} \in \MinSum(\vec{B'})$ there is $\vec{X} \in \MinSum(\vec{B})$ such that $\vec{X} \leq \vec{X'}.$ 

For (iii), by \cref{coro:min-max-elements}, any $\vec{X} \in \MinSum(\vec{B})$ satisfies $\vec{X} \leq \MinSum^1(\vec{B}).$ By \cref{lemma:minsum-resolute-is-cast-monotonic} and \cref{lemma:equival-pointwise-increase} for resolute rules, $\MinSum^1(\vec{B}) \leq \MinSum^1(\vec{B'}).$ By transitivity, $\vec{X} \leq \MinSum^1(\vec{B'}),$ as required.

For (iv), by \cref{coro:min-max-elements}, any $\vec{X'} \in \MinSum(\vec{B'})$ satisfies $\MinSum^0(\vec{B'}) \leq \vec{X'}.$ By \cref{lemma:minsum-resolute-is-cast-monotonic} and \cref{lemma:equival-pointwise-increase} for resolute rules, $\MinSum^0(\vec{B}) \leq \MinSum^0(\vec{B'}).$ By transitivity, $\MinSum^0(\vec{B}) \leq \vec{X'},$ as required.
\end{proof}

\thcastsumfailsandor*
\begin{proof} We show this for $\textsc{Or}_2,$ as the case of $\textsc{And}_2$ is analogous. 
   Consider agents $N = \{a, b, c, d, e, f, \text{\underline{zero}}\}$ with ballots $B_a = c \succ d \succ 1, B_b = \text{\underline{zero}} \succ 1, B_c = a \lor b \succ d \succ 1, B_d = a \lor b \succ c \succ 1, B_e = b \succ 1, B_f = b \succ 1$ and $B_\text{\underline{zero}} = 0.$ First, note that the certificate $\vec{c} = (0, 1, 0, 0, 0, 0, 0)$ of sum 1 is consistent, leading to concrete votes $\vec{X} = (1, 1, 1, 1, 1, 1, 0).$ Moreover, one can check that any consistent certificate $\vec{c'}$ satisfying $c'_b = 0$ requires $c'_a + c'_c + c'_d \geq 2.$ Hence, $\vec{c}$ is the unique \MinSum certificate with corresponding concrete votes $\vec{X},$ from which $\MinSum_\calF(\vec{B}) = \{\vec{X}\}.$

   Now, consider ballots $\vec{B'}$ where the only difference is $a$ switching their vote to $B'_a = 1.$ Here, the unique \MinSum certificate has everyone delegating to their first preferences, leading to the corresponding concrete votes $\vec{X'} = (1, 0, 1, 1, 0, 0, 0),$ from which $\MinSum_\calF(\vec{B'}) = \{\vec{X'}\}.$

   Since $\vec{X} \nleq \vec{X'}$ this fails cast monotonicity by \cref{lemma:equival-pointwise-increase}. For instance, this is the case for the majority aggregation function $\agg.$
   The results concerning $\MinSum^p_\mathcal{F}$ for $p \in \{0, 1\}$ follow since $\MinSum_\calF$ is resolute for $\vec{B}, \vec{B'}$
\end{proof}

\subsection{Axiomatic Dichotomies}

One can combine \cref{thm:min-sum-is-cast-monotonic,th:cast-sum-fails-and-or} with reasoning similar to that in the proof of \cref{lemma:dich-minsum} to get that $\calF = \LIQUID$ is the only class of monotonic functions for which $\MinSum_\mathcal{F}$ and its biased variants satisfy cast monotonicity.\footnote{Negations do not help either:~the beginning of \cref{section:Complexity} gives an example inline.} This gives the following interesting axiomatic dichotomy.

\begin{corollary} Assume $\calF \subseteq \textsc{Mon}.$ If $\calF = \LIQUID,$\footnote{Recall that we implicitly assume $\LIQUID \subseteq \calF.$} then $\MinSum_\calF, \MinSum^0_\calF, \MinSum^1_\calF$ satisfy cast monotonicity. Otherwise, none of them satisfies cast monotonicity.
\end{corollary}

The same dichotomy can also be proven for \LexiMin by noting that the instance constructed in the proof of \cref{th:cast-sum-fails-and-or} still applies.

\subsection{Note on a Previous Hardness Proof} 

It is worth noting that, upon closer inspection, the proof of Colley et al.~(\citeyear{grandi}) that $\MinSum_\textsc{And}$ is NP-hard reduces from Feedback Arc Set. If a variant of Feedback Arc Set where each node has out-degree at most $2$ were to be used, then NP-hardness for $\MinSum_{\textsc{And}_2}$ would follow, but not its hardness of approximation. In the present paper, we get to this result through a different reduction that also enables proving the hardness of approximation. A proof that this variant of Feedback Arc Set is NP-hard is relatively difficult to identify in the literature: we were only able to identify one in unrefereed sources,\footnote{\url{https://cs.stackexchange.com/questions/117291/is-finding-the-minimum-feedback-arc-set-on-graph-with-two-outgoing-arcs-for-each}} although it is mentioned without proof in refereed works.

\end{document}